\UseRawInputEncoding
\documentclass{article}
\usepackage{amsmath, amssymb}\usepackage{extarrows}
\usepackage{amsfonts}
\usepackage{appendix}
\usepackage{graphicx}
\usepackage{xcolor}
\usepackage{url}
\usepackage{bm}
\usepackage[all]{xy}
\usepackage{arydshln}
\usepackage{colortbl}
\usepackage{tabularx}
\usepackage{makecell}
\usepackage{textcomp}

\usepackage{multirow}
\usepackage{diagbox}
\usepackage{threeparttable}
\usepackage{soul}

\newtheorem{fact}{Fact}

\newtheorem{theorem}{Theorem}
\newtheorem{lemma}{Lemma}
\newtheorem{corollary}{Corollary}
\newtheorem{property}{Property}
\newtheorem{definition}{Definition}

\newtheorem{remark}{Remark}

\newenvironment{proof}{{\noindent\it Proof}\quad}{\hfill $\square$\par}
\newtheorem{example}{Example}

\newcommand{\Z}{\ensuremath{\mathbb Z}}

\newcommand{\R}{\ensuremath{\mathbb R}}
\newcommand{\C}{\ensuremath{\mathbb C}}

\newcommand{\ls}[1]
{\dimen0=\fontdimen6\the\font\lineskip=#1\dimen0
	\advance\lineskip.5\fontdimen5\the\font
	\advance\lineskip-\dimen0
	\lineskiplimit=0.9\lineskip
	\baselineskip=\lineskip
	\advance\baselineskip\dimen0
	\normallineskip\lineskip\normallineskiplimit\lineskiplimit
	\normalbaselineskip\baselineskip
	\ignorespaces}

\begin{document}
	
	\bibliographystyle{abbrv}
	
	\title{Non-standard Golay Complementary Sequence Pair over QAM}
\renewcommand{\thefootnote}{}
\footnotetext{
	The authors are supported in part by the National Natural Science Foundation of China under Grant 62172319 and Grant U19B2021; in part by the National Key Research and Development Program under Grant 2021YFA000503.
	The material in this paper was  presented in part at The 10th International Workshop on Signal Design and its Applications in Communications (IWSDA'2022).}
\author{Erzhong Xue$^{1}$, Zilong Wang$^1$,  Guang Gong$^2$\\
	\small $^1$ State Key Laboratory of Integrated Service Networks, Xidian University \\[-0.8ex]
	\small Xi'an, 710071, China\\
	\small $^2$Department of Electrical and Computer Engineering, University of Waterloo \\
	\small Waterloo, Ontario N2L 3G1, Canada  \\
	\small\tt 2524384374@qq.com, zlwang@xidian.edu.cn,  ggong@uwaterloo.ca\\
}
	
	\maketitle

	\thispagestyle{plain} \setcounter{page}{1}

\ls{1.5}

\begin{abstract}
We generalize the three-stage process for constructing and enumerating Golay array and sequence pairs given in 2008 by Frank Fiedler {\em et al.} [A multi-dimensional approach to the construction and enumeration of Golay complementary sequences, Journal of Combinatorial Theory, Series A 115 (2008) 753–776] to $4^{q}$-QAM constellation based on para-unitary matrix method, which partly solves their open questions.
Our work not only includes the main part of known results of Golay complementary sequences over $4^{q}$-QAM based on Boolean functions and standard Golay sequence pairs over QPSK, but also generates new Golay complementary arrays (sequences) over $4^{q}$-QAM based on non-standard Golay array pairs over QPSK.
\end{abstract}

\section{Introduction}

Golay complementary sequence pairs  have found application in many areas of digital information processing since their introduction by Golay \cite{Golay1961Complementary}, 
especially in power control for multicarrier wireless transmission. Based on different modulation, the elements of designed sequences should lie in the different alphabet. 
For example, $H$-ary and  $4^q $-QAM  sequences are employed in phase shift keying (PSK)  modulation and  quadrature amplitude modulation (QAM) respectively.

A milestone work on the construction of  $H$-ary Golay sequence pairs was made   by Davis and Jedwab \cite{Davis1999Peak} in 1999, where an explicit algebraic normal form for  a set of $(H=2^h)$-phase Golay sequence pairs of length $2^m$ was given. Paterson \cite{Paterson00} showed the same construction holds without modification for any even $H$.
Golay sequences of length $2^m$ with this algebraic normal form have 
been referred to as the {\em standard} since then. The {\em non-standard} quaternary Golay sequences were found by computer search \cite{Li05}, and were explained shortly afterwards by  means of \lq\lq cross-over\rq\rq\ of standard Golay sequence pairs with shared auto-correlation function \cite{Fiedler06}. Fiedler, Jedwab and Parker \cite{Array2} introduced the  view of a Golay sequence pair as the projection of a multi-dimensional Golay array pair, and gave a three-stage process to construct Golay sequence pairs, which successfully recover all the standard Golay complementary sequence pairs by using trivial Golay pairs of
length 1 as inputs, and find  new infinite families of Golay  pairs by using non-standard seed Golay pairs as additional inputs. Later on, new source of seed pairs for $6$-ary Golay sequences of length $2^m$ was found and new  families of Golay  pairs by using these Golay pairs as  additional inputs was shown in \cite{newseed}. Thus, one only need to find
$H$-ary seed pairs to construct   $H$-ary Golay  pairs from the view of this three-stage process.

Since QAM are widely employed in the high rate wireless transmissions,  $16$-QAM and $64$-QAM Golay sequence pairs of length $2^{m}$ were constructed in \cite{Robing2001A,Chong2003A,Lee2006A,Li2008Comments,Chang2010New}. Those sequences were generalized to $4^q$-QAM and represented as the weighted-sum of selected standard quaternary Golay sequences in five cases by Li  \cite{Li2010A} and Liu  \cite{Liu2013New}. In addition,  Li~\cite{Li2008Comments} found out that non-standard quaternary Golay sequences can also be combined to generate QAM Golay sequences by providing an example of 64-QAM Golay complementary pair.

Array is an important object in the study of combinatorics.
The three-stage process \cite{Fiedler06} for constructing Golay arrays and sequences is very successful  for PSK case. It is then nature to ask how can it be used to simplify or extend known results on the construction of $4^q$-QAM Golay arrays and sequences. However, this question posed in \cite{Fiedler06} is open until now, since the process to construct higher-dimensional Golay array pairs from lower-dimensional Golay array pairs \cite[Theorems 5 and 7]{Fiedler06} could not be extended from PSK case to QAM case. 

In this paper,  we demonstrate the power of the recently-introduced view \cite{Budi2018PU, fullpaper} of studying Golay array (sequence) pairs by para-unitary (PU)  matrices,  which enable us to construct higher-dimensional Golay array pairs from lower-dimensional Golay array pairs for both PSK and QAM case. Instead of directly studying a Golay array pair, we introduce the concept of a Golay array matrix. A similar three-stage process working for both PSK and QAM was proposed.

For the first stage, we propose the recursive formula for Golay array matrices which is initially in the domain of their generating functions and then converted to Golay array matrices themselves. For the PSK case, this stage is equivalent to the first stage in \cite[Theorems 5 and 7]{Array2}.
Since the recursive formula is presented in the domain of generating functions, we are able to construct higher-dimensional $4^q$-QAM Golay array pairs from lower-dimensional weighted-sum of quaternary Golay array pairs, where the weights are determined by the  factorization of integer $q$. Thus, the first stage here can be generalized to QAM case. For the second stage, we proved the affine offsets can be applied to Golay array matrices for both PSK and QAM case.
For the third stage,
we introduce the mixed radix representation \cite{fullpaper} to simplify the graph-theoretic means \cite[Proposition 2]{Array2} of tracking the effect of successive projections of Golay array pairs from higher dimension  to lower dimensions.

Note that the paper \cite{fullpaper} focuses on the algebraic normal form  of  $4^q$-QAM Golay sequences pairs which are projection of  Golay array pairs of size $2\times 2\times \cdots \times 2$, while this paper give the systematic construction based on the existing Golay sequences pairs including nonstandard pairs. The new three-stage process working for both PSK and QAM was proposed. In addition, we use  this process to recover the main part of known $4^q$-QAM Golay sequences pairs of length $2^m$ with explicit function form \cite{Robing2001A,Chong2003A,Lee2006A,
Li2010A,fullpaper}, and construct infinite family of new  $4^q$-QAM Golay sequences pairs by taking quaternary non-standard seed Golay pairs as inputs. A lower bound of the new constructed Golay pairs over QAM are also discussed.

The rest of this paper is organized as follows. In Section \ref{Sec: Definition and Notations}, we introduced Golay array (sequence) pairs and their generating functions, and extend the concepts to Golay array (sequence) matrix and PU matrices. We also introduced the projection from arrays of higher dimensions to that of lower dimensions (or sequences). We introduce function matrices and affine functions for the $H$-PSK case and $4^{q}$-QAM case.
In Section \ref{Sec: PSK GCP Based on PU}, we present the constructions of Golay array matrices over $H$-PSK and $4^{q}$-QAM based on PU method, and introduced the projected sequences.
Based on that, in Section \ref{Sec: QAM GCP and V-GBF}, we give the specific construction using  weighted sum of compatible Hadamard matrices and cross-over PU matrices. 
In Section \ref{Sec: Enumerations},
a lower bound of enumeration is given.
We conclude the paper in Section \ref{Sec: Concluding Remarks}.

\section{Preliminary}\label{Sec: Definition and Notations}
\subsection{Basic Definition and Notations}

An {\em array} of size ${b}_1\times {b}_2\cdots\times {b}_n$ is an $n$-dimensional matrix of complex-valued entries, which can be expressed by an $n$-variable function
	$$\mathcal{F}(\bm{y})=\mathcal{F}({y}_1,{y}_2,\cdots,{y}_{n}): \mathbf{Z}_{{b}_1}\times \mathbf{Z}_{{b}_2}\cdots\times \mathbf{Z}_{{b}_n}\rightarrow \C,$$
where $\mathbf{Z}_s$ means the set $\{0,1,\dots,b-1\}$, $\C$ means the complex field.
	The {\em generating function}  corresponding to the complex-valued array $\mathcal{F}(\bm{y})$ is the  (multivariate) polynomial
	\begin{equation}\label{mul-var}
		F(\bm{z})=F({z}_1, {z}_2, \cdots {z}_{n})=\sum_{{y}_1,{y}_2,\cdots, {y}_{n}}\mathcal{F}({y}_1,{y}_2,\cdots,{y}_{n})\cdot{z}_1^{{y}_1}{z}_2^{{y}_2}\cdots {z}_{n}^{{y}_{n}}.
	\end{equation}
In particular, the one dimensional array is actually a sequence which is  denoted by $\mathcal{F}(y)$ and its generating function is denoted by $F(z)$ in this paper.

The {\em aperiodic cross-correlation} of two arrays $\mathcal{F}(\bm{y})$ and $\mathcal{G}(\bm{y})$ at shift $\bm{\tau}=(\tau_1, \tau_2, \cdots \tau_{n})$ ($1-{b}_{k}\leq \tau_k\leq {b}_{k}-1$) is defined by
\[C_{\mathcal{F},\mathcal{G}}(\bm{\tau})=
\sum_{\bm{y}}{\mathcal{F}(\bm{y}+\bm{\tau})\cdot \overline{\mathcal{G}}(\bm{y})},\]
where ${F(\bm{y}+\bm{\tau})\cdot \overline{F}(\bm{y})}=0$ if either $F(\bm{y}+\bm{\tau})$ or $\overline{\mathcal{G}}(\bm{y})$ is not defined.	
	The {\em aperiodic auto-correlation} of an array $F(\bm{y})$ at shift $\bm{\tau}$ is denoted by 
	$C_{\mathcal{F}}(\bm{\tau})=C_{\mathcal{F},\mathcal{F}}(\bm{\tau})
	$.
%

\subsection{Golay Array (Sequence) Pair}
Suppose $\left\{\mathcal{F}(\bm{y}),\mathcal{G}(\bm{y})\right\}$ are a pair of $n$-dimensional arrays, whose corresponding generating functions are $\{{F}(\bm{z}),{G}(\bm{z})\}$.
\begin{definition}[Golay array (sequence) pair]\cite{Array2}\label{def: GAP}
The $n$-dimensional arrays $\mathcal{F}(\bm{y})$ and $\mathcal{G}(\bm{y})$  are called  {\em complementary}  if either of the following equivalent conditions in viewpoint from correlation functions and generating functions is satisfied. 
\begin{description}
\item[Corelation function]
\begin{equation}
{C}_{\mathcal{F}}(\bm{\tau})+{C}_{\mathcal{G}}(\bm{\tau})=0,\quad(\forall \bm{\tau}\ne \bm{0}).
\end{equation}
\item[Generating function]
\begin{equation}
F(\bm{z})\cdot \overline{F}(\bm{z}^{-1})+G(\bm{z})\cdot \overline{G}(\bm{z}^{-1})=c
\end{equation}
where $c$ is a real constant.
\end{description}
$\{\mathcal{F}(\bm{y}),\mathcal{G}(\bm{y})\}$  is called a {\em Golay complementary array pair} (or Golay array pair). 
And  either array in a Golay array pair is called a  {\em Golay array}.
In particular, if the dimension $n=1$, the Golay array pair is also called a {\em Golay complementary sequence pair (or Golay sequence pair)} and  either sequence in a Golay sequence pair is called a  {\em Golay complementary sequence (or Golay sequence)}.
\end{definition} 

 For more details on the concepts and results for Golay array pairs, see \cite{Array2, CCA}.

\subsection{Golay Array Matrix and PU Matrix}\label{Subec: Golay array matrix and PU matrix}
For $0\leq i,j \leq 1$, let $\mathcal{F}_{i,j}(\bm{y})$ be $n$-dimensional arrays and ${F}_{i,j}(\bm{z})$ their corresponding generating functions. Denote the {\em array matrix} $\bm{\mathcal{M}}(\bm{y})$ by
\begin{equation}\label{eq: array-matrix}
\bm{\mathcal{M}}({\bm{y}})=
\begin{bmatrix}
\mathcal{F}_{0,0}({\bm{y}})&\mathcal{F}_{0,1}({\bm{y}})\\
\mathcal{F}_{1,0}({\bm{y}})&\mathcal{F}_{1,1}({\bm{y}})\\
\end{bmatrix}.
\end{equation}
Denote their corresponding {\em generating matrix}  $\bm{M}(\bm{z})$ by
\begin{equation}\label{eqn: gene-matrix}
\bm{M}({\bm{z}})=
\begin{bmatrix}
{F}_{0,0}({\bm{z}})&{F}_{0,1}({\bm{z}})\\
{F}_{1,0}({\bm{z}})&{F}_{1,1}({\bm{z}})\\
\end{bmatrix}.
\end{equation}

\begin{definition}[Golay array matrix]\label{def: GAM}\cite[Def 6, Thm 1]{CCA}\label{def: CCA}
$\bm{\mathcal{M}}({\bm{y}})$ in equation (\ref{eq: array-matrix})  is called a {\em Golay array matrix} if 
either of the following equivalent conditions in viewpoint from correlation functions and generating functions  is satisfied.
\begin{description}
\item[Corelation function:] the two array pairs in each columns (and rows) of $\bm{\mathcal{M}}({\bm{y}})$ form Golay array pairs which are orthogonal with each other, i.e.,
\begin{equation}
\left\{
\begin{aligned}
&{C}_{\mathcal{F}_{0,j}}(\bm{\tau})+{C}_{\mathcal{F}_{1,j}}(\bm{\tau})=0,
\quad
{C}_{\mathcal{F}_{i,0}}(\bm{\tau})+{C}_{\mathcal{F}_{i,1}}(\bm{\tau})=0,
\quad(i,j\in\{0,1\};\ \forall \bm{\tau}\ne \bm{0}),\\
&{C}_{\mathcal{F}_{0,0},\mathcal{F}_{0,1}}(\bm{\tau})+{C}_{\mathcal{F}_{1,0},\mathcal{F}_{1,1}}(\bm{\tau})=0,
\quad
{C}_{\mathcal{F}_{0,0},\mathcal{F}_{1,0}}(\bm{\tau})+{C}_{\mathcal{F}_{0,1},\mathcal{F}_{1,1}}(\bm{\tau})=0,
\quad(\forall \bm{\tau}).
\end{aligned}\right.
\label{equation_GCP}	
\end{equation}

\item[Generating function:] $\bm{M}(\bm{z})$ is  a {\em (multivariate) para-unitary} (PU) matrix, i.e.,
\begin{equation}
	\bm{M}(\bm{z})\cdot\bm{M}^{\dagger}(\bm{z}^{-1})=c\cdot \bm{I},
\end{equation}
where  $c$ is a real number, $(\cdot)^\dagger$ denotes the Hermitian transpose, and $\bm{I}$ is the identity matrix of order $2$.
\end{description}
\end{definition}

\begin{remark}
1. The concept of PU matrix is generating function form of Golay array matrix.

2. The Golay array matrix is actually a special case of the {\em complete complementary array} studied in \cite{CCA} for the size $2\times2$.
The elements of arrays studied in \cite{CCA} are $H$-th roots of unity, while in this paper the arrays defined in Definition \ref{def: GAM} lie in complex field.

3. For the one dimensional (sequence) case, the Golay sequence matrix is also called a complementary mate pair in \cite{Tseng1972Complementary}.
\end{remark}

According to Definition \ref{def: CCA}, the array pair in each column (or row) of a Golay array matrix form a Golay array pair.
Conversely, given a Golay array pair, we can always construct a Golay array matrix by following method. 
Suppose $\mathcal{F}(\bm{y})=\mathcal{F}({y}_1,{y}_2,\cdots,{y}_{n})$ is a ${b}_1\times {b}_2\cdots\times {b}_n$ complex-valued array.	Denote  $\mathcal{F}^{*}(\bm{y})$ as its conjugate reverse array, i.e.,
\[\mathcal{F}^{*}(\bm{y})=\mathcal{F}^{*}({y}_1,{y}_2,\cdots,{y}_{n})=\overline{\mathcal{F}}({b}_{1}-1-{y}_1,{b}_{2}-1-{y}_2,\cdots,{b}_{n}-1-{y}_{n}).\]
It is easy to verify that the generating function of $ \mathcal{F}^{*}(\bm{y})$, which denoted by ${F}^{*}(\bm{z})$, is given by
\[{F}^{*}(\bm{z})={F}^{*}({z}_1,{z}_2,\cdots,{z}_{n})={z}_{1}^{{b}_1-1}{z}_{2}^{{b}_2-1}\dots {z}_{n}^{{b}_n-1}\cdot\overline{F}(\bm{z}^{-1}).\]

As \cite{Array2} pointed out, $\mathcal{F}^{*}(\bm{y})$ and $\mathcal{F}(\bm{y})$ share the same aperiodic auto-correlation functions, i.e., ${C}_{\mathcal{F}^{*}}(\bm{\tau})={C}_{\mathcal{F}}(\bm{\tau})$, which is equivalent to ${F}^{*}(\bm{z})\cdot\overline{F}^{*}(\bm{z}^{-1})={F}(\bm{z})\cdot\overline{F}(\bm{z}^{-1})$.
Similarly, it's easy to verify that ${C}_{\mathcal{F}^{*},\mathcal{G}^{*}}(\bm{\tau})={C}_{\mathcal{G},\mathcal{F}}(\bm{\tau})$, which leads to ${F}^{*}(\bm{z})\cdot\overline{G}^{*}(\bm{z}^{-1})={G}(\bm{z})\cdot\overline{F}(\bm{z}^{-1})$. 
Thus we have the following lemma.
\begin{lemma}\label{lem: GAP->GAM}
For a Golay array pair $\{\mathcal{F}(\bm{y}),\mathcal{G}(\bm{y})\}$ of size ${b}_1\times {b}_2\cdots\times {b}_n$, the array matrix given by
\begin{equation}\label{eq: lem_GAP->GAM M(y)}
	\bm{\mathcal{M}}(\bm{y})=\begin{bmatrix}
\mathcal{F}({\bm{y}})&\mathcal{G}^{*}({\bm{y}})\\
\mathcal{G}({\bm{y}})&-\mathcal{F}^{*}({\bm{y}})\\
\end{bmatrix},
\end{equation}
is a Golay array matrix.
The corresponding generating matrix
\begin{equation}\label{eq: lem_GAP->GAM M(z)}
\bm{M}(\bm{z})=\begin{bmatrix}
{F}({\bm{z}})&{G}^{*}({\bm{z}})\\
{G}({\bm{z}})&-{F}^{*}({\bm{z}})\\
\end{bmatrix}, 
\end{equation}
is a PU matrix.
\end{lemma}

In \cite[Theorem 5]{Array2}, Golay array pair $(\mathcal{F}({\bm{y}}), \mathcal{G}({\bm{y}}))$ and it's conjugate reverse $(\mathcal{F}^{*}({\bm{y}}),\mathcal{G}^{*}({\bm{y}}))$ are used to construct suitable Golay array pairs from lower-dimensional Golay array (or sequence) pairs, and is called generalization of Dymond’s construction \cite{Dymond1992}.
The first stage in \cite{Array2} is explicitly via repeated use of this process, whose function form results are presented in \cite[Theorem 7]{Array2}.
The process is equivalent to PU matrices method if we use \eqref{eq: lem_GAP->GAM M(z)} as ``seed" PU matrices. The details of PU matrices method are introduced in Section \ref{Sec: PSK GCP Based on PU}.

\subsection{Projection  from Arrays to that of Lower Dimensions}
In \cite{Array2}, the authors introduced the projection mappings that reduce the dimensions of an array and preserve the complementary property. 
Golay array with low dimensions or Golay sequence could be derived from Golay array with high dimensions by the iteration of the process \cite{Array1}.

In this section we introduce the concept ``Mixed radix representation'' to simplify the graph-theoretic means of tracking the effect of successive projections in \cite{Array2}.

\begin{definition}[Mixed radix system]\label{def:mixed radix}\cite{book_mixed_radix}
	Given radixes $\{{b}_{k}\ge1|1\leq{k}\leq{m}\}$. 
	Let $\pi$ be a permutation of  $\{1,2,\dots,{m}\}$.
	Define the bases $B_{k}$ ($1\leq{k}\leq{m}$) by
	\begin{equation}\label{eq: Bases_B}
	{B}_{\pi(k)}=\prod_{i=1}^{k-1}{b}_{\pi(i)},
	\end{equation}
where ${B}_{\pi(1)}=1$.
For ${y}_{k}\in\mathbf{Z}_{{b}_{k}}$ ($1\leq k\leq{m}$), if
	\begin{equation}\label{eqn: projection}
	y=\sum_{k=1}^{m}{y}_{k}\cdot{B}_{k},
	\end{equation}
then
	 $({y}_{1},{y}_{2},\dots,{y}_{m})$ is called a {\em mixed radix representation} of $y$ based on radixes $({b}_{1},{b}_{2},\dots,{b}_{m})$ and permutation $\pi$.
\end{definition}

%

The sequence $\mathcal{F}({y})$ of length $\prod_{k=1}^{m}{b}_{k}$ is called a {\em projection} from the array $\mathcal{F}({y}_{1},{y}_{2},\dots,{y}_{m})$ of size ${b}_{1}\times{b}_{2}\times\dots\times{b}_{m}$ if
\[
\mathcal{F}'({y})=\mathcal{F}({y}_{1},{y}_{2},\dots,{y}_{m}),
\]
where $({y}_{1},{y}_{2},\dots,{y}_{m})$ is a mixed radix representation of $y$ based on radixes $({b}_{1},{b}_{2},\dots,{b}_{m})$ and permutation $\pi$.
It is straightforward that their corresponding generating functions are connected by
\[
{F}'({z})=
{F}({z}^{B_1},{z}^{B_2},\dots,{z}^{B_m}),
\]
where  ${B_{k}}$ are bases determined by \eqref{eq: Bases_B}.

Similarly, the projection from arrays to that of lower dimensions can be easily expressed via the mixed-radix system as follows.

For $1\leq{v}\leq{n}$, define $\bm{y}_{v}=(y_{v,1},y_{v,2},\dots,y_{v,m_{v}})$.
The array $\mathcal{F}'({y}_{1},{y}_{2},\dots,{y}_{n})$ of size $(\prod_{k_1=1}^{m_1}{b}_{1,k_1})\times (\prod_{k_2=1}^{m_2}{b}_{2,k_2})\times\cdots\times(\prod_{k_n=1}^{m_n}{b}_{n,k_n})$ is called {\em projected} from the array 
$\mathcal{F}(\bm{y}_{1},\bm{y}_{2},\dots,\bm{y}_{n})$ 
of size ${b}_{1,1}\times{b}_{1,2}\times\dots\times{b}_{1,m_1}\times\dots\times{b}_{n,1}\times{b}_{n,2}\times\dots\times{b}_{n,m_n}$ 
if
\[
\mathcal{F}'({y}_{1},{y}_{2},\dots,{y}_{n})=\mathcal{F}(\bm{y}_{1},\bm{y}_{2},\dots,\bm{y}_{n}),
\]
where for all $1\leq{v}\leq{n}$, $({y}_{v,1},{y}_{v,2},\dots,{y}_{v,m_v})$ is a mixed radix representation of $y_{v}$ based on radixes $({b}_{v,1},{b}_{v,2},\dots,{b}_{v,m_v})$ and permutation $\pi_{v}$ of $\{1,2,\dots,{m}_{v}\}$.
It is straightforward that their corresponding generating functions are connected by
\[
{F}'({z}_{1},{z}_{2},\dots,{z}_{n})={F}({z}_{1}^{B_{1,1}},{z}_{1}^{B_{1,2}},\dots,{z}_{1}^{B_{1,m_1}},\dots, {z}_{n}^{B_{n,1}},{z}_{n}^{B_{n,2}},\dots,{z}_{n}^{B_{n,m_n}}),
\]
where $B_{v,k_{v}}$ ($1\leq{k}_{v}\leq{m}_{v}$) is determined by ${B}_{v,\pi_{v}(k_{v})}=\prod_{i=1}^{k_{v}-1}{b}_{v,\pi_{v}(i)}$.


If the projection is applied to Golay array pair and Golay array matrix, we have the following Property.
\begin{property}\label{pro: projection}
	\begin{enumerate}
		\item\cite{Array2,Array1}
		The array pair $\{\mathcal{F}({y}_{1},{y}_{2},\dots,{y}_{n}),\mathcal{G}({y}_{1},{y}_{2},\dots,{y}_{n})\}$  projected from the Golay array pair $\{\mathcal{F}(\bm{y}_{1},\bm{y}_{2},\dots,\bm{y}_{n}),\mathcal{G}(\bm{y}_{1},\bm{y}_{2},\dots,\bm{y}_{n})\}$ form a Golay array pair.
		\item
		\label{prop-2}
		The array matrix \[\bm{\mathcal{M}}({y}_{1},{y}_{2},\dots,{y}_{n})=\left[\begin{matrix}
			\mathcal{F}_{0,0}({y}_{1},{y}_{2},\dots,{y}_{n})&\mathcal{F}_{0,1}({y}_{1},{y}_{2},\dots,{y}_{n})\\
			\mathcal{F}_{1,0}({y}_{1},{y}_{2},\dots,{y}_{n})&\mathcal{F}_{1,1}({y}_{1},{y}_{2},\dots,{y}_{n})\\
		\end{matrix}\right]\]
	projected from the Golay array matrix \[\bm{\mathcal{M}}(\bm{y}_{1},\bm{y}_{2},\dots,\bm{y}_{n})=\left[\begin{matrix}
			\mathcal{F}_{0,0}(\bm{y}_{1},\bm{y}_{2},\dots,\bm{y}_{n})&\mathcal{F}_{0,1}(\bm{y}_{1},\bm{y}_{2},\dots,\bm{y}_{n})\\
			\mathcal{F}_{1,0}(\bm{y}_{1},\bm{y}_{2},\dots,\bm{y}_{n})&\mathcal{F}_{1,1}(\bm{y}_{1},\bm{y}_{2},\dots,\bm{y}_{n})\\
		\end{matrix}\right]\]
	form a Golay array matrix.
	\end{enumerate}
\end{property}
\begin{proof}
	We shall give the proof of the second part. The first proposition can be obtained similarly and is proved in \cite{Array2}.
	
	According to Definition \ref{def: CCA}, 
	the generating matrix of Golay array matrix $\bm{\mathcal{M}}(\bm{y}_{1},\bm{y}_{2},\dots,\bm{y}_{n})$ satisfy that $\bm{M}(\bm{z}_{1},\bm{z}_{2},\dots,\bm{z}_{n})\cdot\bm{M}^{\dagger}(\bm{z}_{1}^{-1},\bm{z}_{2}^{-1},\dots,\bm{z}_{n}^{-1})=c\cdot \bm{I}$.
	Since $\bm{M}({z}_{1},{z}_{2},\dots,{z}_{n})$, the generating matrix of Golay array matrix $\bm{\mathcal{M}}({y}_{1},{y}_{2},\dots,{y}_{n})$, can be derived from $\bm{M}(\bm{z}_{1},\bm{z}_{2},\dots,\bm{z}_{n})$.
	It is easy to known $\bm{M}({z}_{1},{z}_{2},\dots,{z}_{n})\cdot\bm{M}^{\dagger}({z}_{1}^{-1},{z}_{2}^{-1},\dots,{z}_{n}^{-1})=c\cdot \bm{I}$.
	This implies that $\bm{\mathcal{M}}({y}_{1},{y}_{2},\dots,{y}_{n})$ is also a Golay array matrix. 
\end{proof}
\begin{remark}
The second part is partly given in \cite[Prop 1]{CCA} for the projection from arrays to sequences.
\end{remark}

The relationships of multivariate (univariate) PU matrices,  Golay array (sequence) matrix and Golay array (sequence) pairs are displayed in Figure \ref{fig-1}. 

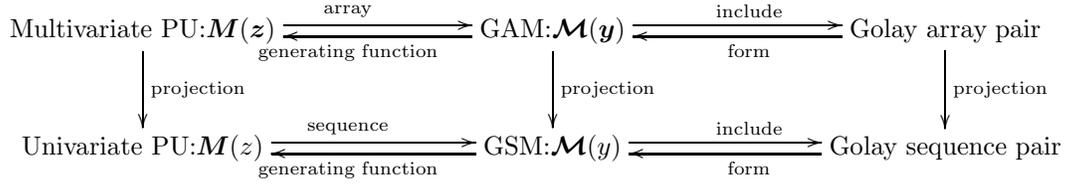
\begin{figure}
	\centering
	$ \begin{gathered}\xymatrix@C=2.5cm@R=1cm{
		\text{Multivariate PU:}{\bm{M}(\bm{z})}
		\ar@<.5ex>[r]^{\text{array}} 
		\ar[d]^{\text{projection}}
		&\text{GAM:}{\bm{\mathcal{M}}(\bm{y})}
		\ar@<.5ex>[l]^{\text{generating function}}
		\ar@<.5ex>[r]^{\text{include}} 
		\ar[d]^{\text{projection}}
		&\ar@<.5ex>[l]^{\text{form}}{\text{Golay array pair}}\ar[d]^{\text{projection}}\\
		\text{Univariate PU:}{\bm{M}({z})}\ar@<.5ex>[r]^{\text{sequence}}  &\text{GSM:}{\bm{\mathcal{M}}({y})}
		\ar@<.5ex>[l]^{\text{generating function}} \ar@<.5ex>[r]^{\text{include}}& \ar@<.5ex>[l]^{\text{form}}{\text{Golay sequence pair}} \\
	}
	\end{gathered} $
	\caption{The relationship of the main concepts}\label{fig-1}
\end{figure}

\subsection{Function Matrices and Affine Functions}

Usually the in-range entries of $ \mathcal{F}(\bm{y}) $ are constrained to lie in a small finite set ${S}$ called the array alphabet. 

Let $\zeta$ denote the ${H}$-th primitive unit root exp$(2\pi\sqrt{-1}/{H})$ for certain even integer ${H}$. If the entries lie in ${S}_{{H}\text{-PSK}} =\{1,\zeta,\zeta^{2},\dots,\zeta^{{H}-1}\}$, $\mathcal{F}(\bm{y})$ is called an ${H}$-PSK array.
In particular, QPSK constellation is denoted by ${S}_{QPSK}={S}_{4\text{-PSK}}$.
The function $f(\bm{y})=f({y}_1,{y}_2,\cdots, {y}_{n}):\mathbf{Z}_{{b}_1}\times \mathbf{Z}_{{b}_2}\cdots\times \mathbf{Z}_{{b}_n}\rightarrow \Z_{H}$ (where $\Z_{H}$ means the ring of integers modulo ${H}$) is called {\em function} of array $\mathcal{F}(\bm{y})$ if
\begin{equation}
	\mathcal{F}({y}_1,{y}_2,\cdots,{y}_{n})= {\zeta}^{f({y}_1,{y}_2,\cdots, {y}_{n})}.
\end{equation}

In particular, if the ${H}$-PSK array is of size ${b}_1\times {b}_2\cdots\times {b}_n=2\times2\times\cdots\times2$ (or abbreviated to size $\bm{2}^{(n)}$), it's function
$f({y}_1,{y}_2,\cdots, {y}_{n})$ can be expressed by the {\em Generalized Boolean function (GBF)} 
$f(x_1, x_2, \cdots, x_{n}):\mathbf{Z}_{2}^{n}\rightarrow\Z_{H}$. The algebraic normal form of GBFs can be uniquely expressed as a linear combination over $\mathbb{Z}_{H}$ of the monomials \[1, x_1, x_2, \cdots, x_{m}, x_1x_2, x_1x_3, \cdots, x_{m-1}x_{m}, \cdots, x_1x_2x_3\cdots x_{m}.\]

Denote the addition of sets by operator $\oplus$, i.e., $a\cdot{S}_{1}\oplus b\cdot{S}_{2}=\{a\alpha+b \beta|\alpha\in{S}_{1},\beta\in{S}_{2}\}$.
The $4^q$-QAM alphabet can be viewed as the weighted sums of $q$  QPSK constellation, with weights in the ratio of $2^{q-1}:2^{q-2}:\dots:1$ \cite{Li2010A}, i.e., ${S}_{4^{q}\text{-QAM}} =\left\{{\bigoplus}_{p=0}^{q-1}2^{p}\cdot{S}_{\text{QPSK}}\right\}$.
Conversely,  the QPSK constellation can be seen as a special case of $4^q$-QAM constellation when $q=1$.	
Let $\xi$ denote the fourth primitive unit root throughout the paper.  
Thus a $4^{q}$-QAM array $F(\bm{y})$ can be viewed as the weighted sums of arrays $\{F^{(p)}(\bm{y})|0\leq p<q\}$ over QPSK and can be uniquely described by multi-output function 
$\vec{f}(\bm{y})
=\vec{f}({y}_1,{y}_2,\cdots,{y}_{n})=\{f^{(p)}(\bm{y})|0\le p<q \}:\mathbf{Z}_{{b}_1}\times \mathbf{Z}_{{b}_2}\cdots\times \mathbf{Z}_{{b}_n}\rightarrow \Z_{4}^{q}$ as follows.
\begin{equation}
	\mathcal{F}(\bm{y})=\sum_{p=0}^{q-1}2^{q-1-p}\cdot \mathcal{F}^{(p)}(\bm{y})=\sum_{p=0}^{q-1}2^{q-1-p}\cdot {\xi}^{f^{(p)}(\bm{y})}.
\end{equation}

\begin{example}
	For example, suppose an $2\times2\times3$ array $\mathcal{F}(y_{1},y_{2},y_{3})$ over $16$-QAM is given by
\begin{equation}
\begin{array}{lll}
	\mathcal{F}(0,0,0)=2i+1, &
\mathcal{F}(0,0,1)=-1,&
\mathcal{F}(0,0,2)=-2i+1,\\
\mathcal{F}(0,1,0)=2i+1,&
\mathcal{F}(0,1,1)=-2+i,&
\mathcal{F}(0,1,2)=-2i-1,\\
\mathcal{F}(1,0,0)=i,&
\mathcal{F}(1,0,1)=-2-i,&
\mathcal{F}(1,0,2)=-3i,\\
\mathcal{F}(1,1,0)=-2-i,&
\mathcal{F}(1,1,1)=-2i+1,&
\mathcal{F}(1,1,2)=3.
\end{array}
\end{equation}
It can be expressed by weighted sum of arrays over QPSK, i.e, $\mathcal{F}(y_{1},y_{2},y_{3})=2\mathcal{F}^{(0)}(y_{1},y_{2},y_{3})+\mathcal{F}^{(1)}(y_{1},y_{2},y_{3})$, where 
\[\left(\mathcal{F}^{(0)}(y_{1},y_{2},y_{3})\middle|_{y_{1},y_{2},y_{3}=(0,0,0)}^{(1,1,2)}\right)=(i,-1,-i,i,-1,-i,i,-1,-i,-1,-i,1),\]
\[\left(\mathcal{F}^{(1)}(y_{1},y_{2},y_{3})\middle|_{y_{1},y_{2},y_{3}=(0,0,0)}^{(1,1,2)}\right)=(1,1,1,1,i,-1,-i,-i,-i,-i,1,i).\]
And the corresponding functions are ${f}^{(0)}(y_{1},y_{2},y_{3})=y_{1}y_{2}+y_{3}+1$, and ${f}^{(1)}(y_{1},y_{2},y_{3})=3y_{1}+y_{2}y_{3}$.
\end{example}

In particular, if the array size is ${b}_1\times {b}_2\cdots\times {b}_n=\bm{2}^{(n)}$,
$\vec{f}(\bm{y})$ can be expressed by multi-output GBF  $\vec{f}(\bm{x})
=\{f^{(p)}(\bm{x})|0\le p<q \}: \mathbf{Z}_{2}^{n}\rightarrow \Z_{4}^{q}$.

In Subsection \ref{Subec: Golay array matrix and PU matrix}, we present the array matrix $\bm{\mathcal{M}}(\bm{y})$ in \eqref{eq: array-matrix} and corresponding generating matrix $\bm{M}(\bm{z})$ in \eqref{eqn: gene-matrix}.
Here we will introduce their corresponding function matrix by
\begin{description}
\item [${H}$-PSK Case :]
\begin{equation}\label{eq: function-matrix}
	\widetilde{{\bm{M}}}({\bm{y}})=
	\begin{bmatrix}
		{f}_{0,0}({\bm{y}})&{f}_{0,1}({\bm{y}})\\
		{f}_{1,0}({\bm{y}})&{f}_{1,1}({\bm{y}})\\
	\end{bmatrix}
	=\left\{{f}_{i,j}(\bm{x},\bm{y})\middle|0\leq{i,j}\leq1\right\}
\end{equation}
where $f_{i,j}(\bm{y})$ over $\Z_{H}$ satisfy $F_{i,j}(\bm{y})={\zeta}^{f_{i,j}(\bm{y})}$ ($0\leq{i},{j}\leq1$), or
\item [$4^{q}$-QAM Case :]
\begin{equation}\label{eqn: V-F-matrix}
	\widetilde{\mathbb{M}}({\bm{y}})=	\begin{bmatrix}
		\vec{f}_{0,0}({\bm{y}})&\vec{f}_{0,1}({\bm{y}})\\
		\vec{f}_{1,0}({\bm{y}})&\vec{f}_{1,1}({\bm{y}})\\
	\end{bmatrix}
	=\left\{{f}_{i,j}^{(p)}(\bm{x},\bm{y})\middle|0\leq{i,j}\leq1;0\leq{p}<q\right\}
\end{equation}
where $\vec{f}_{i,j}(\bm{y})
=\{f_{i,j}^{(p)}(\bm{y})|0\le p<q \}$, 
$f_{i,j}^{(p)}(\bm{y})$ over $\Z_{4}$ satisfy $\mathcal{F}_{i,j}(\bm{y})=\sum_{p=0}^{q-1}2^{q-1-p}\cdot {\xi}^{f_{i,j}^{(p)}(\bm{y})}$ ($0\leq{i},{j}\leq1$). 
\end{description}

An affine function ${l}(\bm{y})$ over $\Z_{H}$  
can be expressed by $\sum_{k=1}^{n}c_k{y}_k+c_{0}$ ($c_k\in \Z_{H},0\le k\le n$), where the additions and multiplications is computed modulo ${H}$.
\begin{property}[Affine function]\label{prop: QAM_linear}
Suppose ${l}(\bm{y})$ is an affine function over $\Z_{H}$.
If $\bm{\mathcal{M}}({\bm{y}})$ in equation (\ref{eq: array-matrix})  is a {Golay array matrix}, then the array matrix $\bm{\mathcal{M}}'({\bm{y}})$ whose  function matrix given by 
		\begin{equation}\label{eq: PSK_linear}
			\widetilde{\bm{M}}'({\bm{y}})=\widetilde{\bm{M}}({\bm{y}})+{l}(\bm{y})\cdot\bm{J}
		\end{equation}
	for ${H}$-PSK case or
		\begin{equation}
			\widetilde{\mathbb{M}}'({\bm{y}})=\widetilde{\mathbb{M}}({\bm{y}})+{l}(\bm{y})\cdot\vec{1}_{q}\cdot\bm{J}
		\end{equation}
	(where ${H}=4$) for $4^{q}$-QAM case, where $\bm{J}$ is the all ``1" matrix of order $2$, is also a Golay array matrix.
	\end{property}
	\begin{proof}
	We shall present the proof for the ${H}$-PSK Case. The proof of $4^{q}$-QAM case can be obtained similarly.
	
	Suppose $F({z}_{1},{z}_{2},\dots,{z}_{n})$ is the generating matrix of array $f({y}_{1},{y}_{2},\dots,{y}_{n})$, i.e,
	\[
	F({z}_{1},{z}_{1},\dots,{z}_{n})=\zeta^{f({y}_{1},{y}_{2},\dots,{y}_{n})}\cdot{z}_1^{{y}_1}{z}_2^{{y}_2}\cdots {z}_{n}^{{y}_{n}}.
	\]
	Then the generating function of $f'({y}_{1},{y}_{2},\dots,{y}_{n})=f({y}_{1},{y}_{2},\dots,{y}_{n})+l({y}_{1},{y}_{2},\dots,{y}_{n})$ is given by
	\[
\begin{split}
	F'({z}_{1},{z}_{1},\dots,{z}_{n})&=\zeta^{f({y}_{1},{y}_{2},\dots,{y}_{n})+\sum_{k=1}^{n}c_k{y}_k+c_{0}}\cdot{z}_1^{{y}_1}{z}_2^{{y}_2}\cdots {z}_{n}^{{y}_{n}}\\
	&=\zeta^{c_{0}}\cdot\zeta^{f({y}_{1},{y}_{2},\dots,{y}_{n})}\cdot(\zeta^{c_{1}}{z}_1)^{{y}_1}(\zeta^{c_{2}}{z}_2)^{{y}_2}\cdots (\zeta^{c_{n}}{z}_n)^{{y}_{n}}\\
	&={\xi}^{c_0}\cdot{F}({\xi}^{c_1}{z}_1, {\xi}^{c_2}{z}_2, \cdots, {\xi}^{c_n}{z}_{n}),
\end{split}
\]		
Suppose $\bm{M}'(\bm{z})$ is the generating matrix of ${\bm{M}}'({\bm{y}})$, 
it is straightforward that
    \begin{equation}
    \bm{M}'(\bm{z})={\xi}^{c_0}\cdot\bm{M}({\xi}^{c_1}{z}_1, {\xi}^{c_2}{z}_2, \cdots, {\xi}^{c_n}{z}_{n}),
    \end{equation}	
    which is a PU matrix since $\bm{M}(\bm{z})$ is a PU matrix.
	This complete the proof.
	\end{proof}

\begin{remark}
	The affine offsets property for PSK case of size ${p}\times{p}\times\dots\times{p}$ (where $p$ is an integer) is introduced in \cite[Property 2, Theorem 1]{CCA}.
We generalize the property to the generalized sizes for PSK case and $4^{q}$-QAM case here.
\end{remark}

\section{Constructions}\label{Sec: PSK GCP Based on PU}

As discussed in Section \ref{Sec: Definition and Notations}, 
the construction of Golay sequence pair can be derived from Golay array pair by projection based on Properties \ref{pro: projection} and \ref{prop: QAM_linear}. On the other hand, the construction of Golay array pair is equivalent to construction of Golay array matrix based on multivariate PU matrix.
Our construction is inspired by the relationship diagram from PU matrix to Golay sequence pair in Figure \ref{fig-1}.

Before the construction, we first introduce the following notations.

For $0\leq k\leq n$, $1\leq i\leq r_{k}$, define variables  ${y}_{k,i}\in\mathbf{Z}_{{b}_{k,i}}$, which measure the coordinates of arrays in the dimensions of size ${b}_{k,i}$. Define ${z}_{k,i}$ as the corresponding indeterminates in the generating functions.
For $1\leq k\leq n$, define variables  ${x}_{k}\in\mathbf{Z}_{2}$, which measure the the coordinates of arrays in the dimensions of size $2$. Define ${z}_{k}$ as the corresponding indeterminates in the generating functions.

Abbreviate variable vector $({y}_{k,1},{y}_{k,1},\dots,{y}_{k,r_k})$ to $\bm{y}_{k}$, size ${b}_{k,1}\times{b}_{k,1}\times\dots\times{b}_{k,r_k}$ to $\bm{b}_{k}$
and indeterminate vector $({z}_{k,1},{z}_{k,1},\dots,{z}_{k,r_k})$ to $\bm{z}_{k}$.
Thus any array $\mathcal{F}(\bm{y}_{k})$ 
in this section are of size 
$\bm{b}_k$,
any function ${f}(\bm{y}_{k})$ are from $\mathbf{Z}_{\bm{b}_{k}}:=\mathbf{Z}_{{b}_{k,1}}\times \mathbf{Z}_{{b}_{k,2}}\cdots\times \mathbf{Z}_{{b}_{k,r_k}}$  to $\Z_{H}$. Polynomial $F(\bm{z}_{k})$ are sum of the multivariate monomials $\bm{z}_{k}^{\bm{y}_{k}}$ ($\bm{y}_{k}\in\mathbf{Z}_{\bm{b}_{k}}$) with weights $\mathcal{F}(\bm{x}_{k})$.
%

Abbreviate 
$({x}_{1},{x}_{2},\dots,{x}_{n},\bm{y}_{0},\bm{y}_{1},\bm{y}_{2},\dots,\bm{y}_{n})$ to $(\bm{x},\bm{y})$ and $({z}_{1},{z}_{2},\dots,{z}_{n},\bm{z}_{0},\bm{z}_{1},\bm{z}_{2},\dots,\bm{z}_{n})$ to $(\bm{z},\bm{z}')$.
Thus any constructed array $\mathcal{F}(\bm{x},\bm{y})$ followed by variable $(\bm{x},\bm{y})$ in this section are of size 
$\bm{2}^{(n)}\times \bm{b}_0\times \bm{b}_1\cdots\times \bm{b}_n$,
any function ${f}(\bm{x},\bm{y})$ are from $(\mathbf{Z}_{2})^{n}\times\mathbf{Z}_{\bm{b}_{0}}\times \mathbf{Z}_{\bm{b}_{1}}\cdots\times \mathbf{Z}_{\bm{b}_{n}}$ to $\Z_{H}$.
Polynomial $F(\bm{z},\bm{z}')$ are sum of the multivariate monomials $\bm{z}^{\bm{x}}\cdot{\bm{z}'}^{\bm{y}}$ 
with weights $\mathcal{F}(\bm{z},\bm{z}')$.
 
Denote {\em delay matrix} $\bm{D}(z)$ by
\begin{equation}
	\bm{D}(z)=\begin{bmatrix}
		1&0\\
		0&z
	\end{bmatrix}.
\end{equation}
For $0\leq{k}\leq{n}$, suppose that
${\bm{U}}^{\{k\}}(\bm{z}_{k})$ are 
generating matrices of array matrices $\bm{\mathcal{U}}^{\{k\}}(\bm{y}_{k})$ of size $\bm{b}_{k}$. Define 
\begin{equation}\label{eq: M(z)}
	\begin{aligned}
		\bm{M}(\bm{z},\bm{z}')
		&=\bm{U}^{\{0\}}(\bm{z}_{0})\cdot\bm{D}({z}_{1})\cdot\bm{U}^{\{1\}}(\bm{z}_{1})\cdot\bm{D}({z}_{2})\cdots\bm{U}^{\{n-1\}}(\bm{z}_{n-1})\cdot\bm{D}({z}_{n})\cdot\bm{U}^{\{n\}}(\bm{z}_{n})
		\\&
		=\bm{U}^{\{0\}}(\bm{z}_{0})\cdot\prod_{k=1}^{n}\left(\bm{D}({z}_{k})\cdot\bm{U}^{\{k\}}(\bm{z}_{k})\right).
	\end{aligned}
\end{equation}	
If ${\bm{U}}^{\{k\}}(\bm{z}_{k})$ are PU matrices, then $ \bm{M}(\bm{z},\bm{z}') $ given in \eqref{eq: M(z)}
is also a PU matrix. 
Thus \eqref{eq: M(z)} generate Golay array matrices from Golay array matrices of lower dimensions.
Denote the corresponding Golay array matrix of $\bm{M}(\bm{z},\bm{z}')$ by $\bm{\mathcal{M}}(\bm{x},\bm{y})$.
If ${\bm{U}}^{\{k\}}(\bm{z}_{k})$ ($0\leq{k}\leq{n}$) are over ${H}$-PSK, then $\bm{M}(\bm{z},\bm{z}')$ is also over ${H}$-PSK.
%
For an aimed constellation $4^{q}$-QAM, we propose a method of choosing proper ${\bm{U}}^{\{k\}}(\bm{z}_{k})$. 
The details are given in the following subsections respectively.

\subsection{Golay Array Matrix over ${H}$-PSK}

\begin{lemma}\cite[Theorem 7]{fullpaper} \label{lem: fun_UDU}
Suppose that ${\bm{U}}^{\{k\}}(\bm{z}_{k})$ in \eqref{eq: M(z)} are generating matrices over ${H}$-PSK constellation, and $\widetilde{\bm{U}}^{\{k\}}(\bm{x}_{k})$ are corresponding function matrices.
The corresponding  function matrix of $ \bm{M}(\bm{z},\bm{z}') $ is given by
\begin{equation}\label{QPSK-GBF-matrix}
	\bm{\widetilde{M}}(\bm{x},\bm{y})
	=\bm{\widetilde{U}}^{\{0\}}(\bm{y}_{0})\cdot\bm{\delta}({x}_{1})\cdot\bm{J}+\sum_{k=1}^{n-1}\bm{J}\cdot
	\bm{\delta}({x}_{k})\cdot\bm{\widetilde{U}}^{\{k\}}(\bm{y}_{k})\cdot\bm{\delta}({x}_{k+1})\cdot\bm{J}+\bm{J}\cdot
	\bm{\delta}({x}_{n})\cdot\bm{\widetilde{U}}^{\{n\}}(\bm{y}_{n}).
\end{equation}
where 
\begin{equation}
	\bm{\delta}(x)=\begin{bmatrix}1{-}x&0\\
		0& x \end{bmatrix}.
\end{equation}
\end{lemma}

Based on Lemma \ref{lem: fun_UDU}, we can abstract the function expression of a Golay array matrix over ${H}$-PSK easily.
\begin{theorem}[PSK Golay array matrix]\label{thm: PSK_CCA}
Suppose that $\bm{\mathcal{U}}^{\{k\}}(\bm{y}_{k})$ are over ${H}$-PSK. The $(i,j)$-th entries in the corresponding function matrix $\widetilde{\bm{U}}^{\{k\}}(\bm{y}_{k})$ is ${f}_{i,j}^{\{k\}}(\bm{y}_{k})$.
$\bm{M}(\bm{z},\bm{z}')$ are over $4^{q}$-QAM,
whose function matrix $\widetilde{{\bm{M}}}(\bm{x},\bm{y})=\left\{{f}_{i,j}(\bm{x},\bm{y})\middle|0\leq{i,j}\leq1\right\}$ is  given by
\begin{equation}\label{eq: non-standard CCA}
{f}_{i,j}(\bm{x},\bm{y})=\sum_{k=0}^{n}
\left({f}_{0,0}^{\{k\}}(\bm{y}_{k})\cdot\bar{x}_{k}\bar{x}_{k+1}
+{f}_{1,0}^{\{k\}}(\bm{y}_{k})\cdot{x}_{k}\bar{x}_{k+1}
+{f}_{0,1}^{\{k\}}(\bm{y}_{k})\cdot\bar{x}_{k}{x}_{k+1}
+{f}_{1,1}^{\{k\}}(\bm{y}_{k})\cdot{x}_{k}{x}_{k+1}\right)
\end{equation}
where $x_{0}=i$ and $x_{n+1}=j$ for simplicity, and $\bar{x}:=1-x$
\end{theorem}
	
\begin{proof}
It is easy to verify that
\begin{subequations}\label{eq: JDUDJ=?}
	\begin{alignat}{1}
		\bm{\widetilde{U}}^{\{0\}}(\bm{y}_{0})\cdot\bm{\delta}({x}_{1})\cdot\bm{J}
		&=\begin{bmatrix}
			{f}_{0,0}^{\{0\}}(\bm{y}_{0})\cdot\bar{x}_{1}+{f}_{0,1}^{\{0\}}(\bm{y}_{0})\cdot{x}_{1}&{f}_{0,0}^{\{0\}}(\bm{y}_{0})\cdot\bar{x}_{1}+{f}_{0,1}^{\{0\}}(\bm{y}_{0})\cdot{x}_{1}\\
			{f}_{1,0}^{\{0\}}(\bm{y}_{0})\cdot\bar{x}_{1}+{f}_{1,1}^{\{0\}}(\bm{y}_{0})\cdot{x}_{1}&{f}_{1,0}^{\{0\}}(\bm{y}_{0})\cdot\bar{x}_{1}+{f}_{1,1}^{\{0\}}(\bm{y}_{0})\cdot{x}_{1}
		\end{bmatrix}\\
		\bm{J}\cdot
		\bm{\delta}({x}_{n})\cdot\bm{\widetilde{U}}^{\{n\}}(\bm{y}_{n})
		&=\begin{bmatrix}
			{f}_{0,0}^{\{n\}}(\bm{y}_{0})\cdot\bar{x}_{n}{+}{f}_{1,0}^{\{n\}}(\bm{y}_{0})\cdot{x}_{n}&{f}_{0,1}^{\{n\}}(\bm{y}_{0})\cdot\bar{x}_{n}{+}{f}_{1,1}^{\{n\}}(\bm{y}_{0})\cdot{x}_{n}\\
			{f}_{0,0}^{\{n\}}(\bm{y}_{0})\cdot\bar{x}_{n}{+}{f}_{1,0}^{\{n\}}(\bm{y}_{0})\cdot{x}_{n}&{f}_{0,1}^{\{n\}}(\bm{y}_{0})\cdot\bar{x}_{n}{+}{f}_{1,1}^{\{n\}}(\bm{y}_{0})\cdot{x}_{n}
		\end{bmatrix}\\
		\bm{J}\cdot
		\bm{\delta}({x}_{1})\cdot\bm{\widetilde{U}}^{\{k\}}(\bm{y}_{k})\cdot\bm{\delta}({x}_{2})\cdot\bm{J}&=\left(
		{f}_{0,0}^{\{k\}}(\bm{y}_{k})\cdot\bar{x}_{k}\bar{x}_{k+1}
		+{f}_{1,0}^{\{k\}}(\bm{y}_{k})\cdot{x}_{k}\bar{x}_{k+1}\right.\notag\\
		&+\left.{f}_{0,1}^{\{k\}}(\bm{y}_{k})\cdot\bar{x}_{k}{x}_{k+1}
		+{f}_{1,1}^{\{k\}}(\bm{y}_{k})\cdot{x}_{k}{x}_{k+1}\right)\cdot\bm{J}.
	\end{alignat}
\end{subequations}	
The final results in \eqref{eq: non-standard CCA} can be obtained by substituting the items in \eqref{QPSK-GBF-matrix} with \eqref{eq: JDUDJ=?}. 
\end{proof}

\begin{remark}
As we discussed in Subsection \ref{Subec: Golay array matrix and PU matrix}, the column (or row) in the Golay array (sequence) matrix form a Golay array (sequence) pair. The result is equivalent to \cite[Theorem 7]{Array2}.
\end{remark}

\subsection{Golay Array Matrix over $4^{q}$-QAM}\label{Sec: GAM-QAM}

We shall give a way of choosing the component PU matrices $\bm{U}^{\{k\}}(\bm{z}_{k})$ in \eqref{eq: M(z)}
to ensure $\bm{M}(\bm{z},\bm{z}')$ in \eqref{eq: M(z)} is over $4^{q}$-QAM, and  $\bm{M}(\bm{z},\bm{z}')$ can be expressed as weighted sum of generating matrices $\bm{M}^{(p)}(\bm{z},\bm{z}')$ over QPSK which are still of the form \eqref{eq: M(z)}. Thus according to Lemma \ref{lem: fun_UDU}, we can extract the function expression of $\bm{M}^{(p)}(\bm{z},\bm{z}')$.

A direct way is to choose one of the component PU matrices $\bm{U}^{\{k\}}(\bm{z}_{k})$ in \eqref{eq: M(z)} to be over ${S}_{4^{q}\text{-QAM}} =\left\{{\bigoplus}_{p=0}^{q-1}2^{p}\cdot{S}_{\text{QPSK}}\right\}$. 
For example ${S}_{16\text{-QAM}} ={S}_{\text{QPSK}}\oplus2\cdot{S}_{\text{QPSK}}
=\left\{1,i,-1,-i\right\}+\left\{2,2i,-2,-2i\right\}$	

The following is an example for $q=2$. 
\begin{example}
Define a PU matrix over $16$-QAM
\begin{equation}
	\bm{U}^{\{1\}}(\bm{z}_{1})=
	\begin{bmatrix}
		+3-i{z}_{1,1}+3i{z}_{1,2}-{z}_{1,1}{z}_{1,2}&-1-3i{z}_{1,1}-i{z}_{1,2}-3{z}_{1,1}{z}_{1,2}\\
		+3-i{z}_{1,1}-3i{z}_{1,2}+{z}_{1,1}{z}_{1,2}&-1-3i{z}_{1,1}+i{z}_{1,2}+3{z}_{1,1}{z}_{1,2}
	\end{bmatrix}.
\end{equation}	
$\bm{U}^{\{1\}}(\bm{z}_{1})$ can be presented by the weighted sum of generating matrices over QPSK by
$ \bm{U}^{\{1\}}(\bm{z}_{1})=2\cdot\bm{U}^{\{1,0\}}(\bm{z}_{1})+2^{0}\cdot\bm{U}^{\{1,1\}}(\bm{z}_{1})$, where
\begin{equation}
	\bm{U}^{\{1,0\}}(\bm{z}_{1})=
	\begin{bmatrix}
		+1-i{z}_{1,1}+i{z}_{1,2}-{z}_{1,1}{z}_{1,2}&-1-i{z}_{1,1}-i{z}_{1,2}-{z}_{1,1}{z}_{1,2}\\
		+1-i{z}_{1,1}-i{z}_{1,2}+{z}_{1,1}{z}_{1,2}&-1-i{z}_{1,1}+i{z}_{1,2}+{z}_{1,1}{z}_{1,2}
	\end{bmatrix},
\end{equation}
\begin{equation}
	\bm{U}^{\{1,1\}}(\bm{z}_{1})=
	\begin{bmatrix}
		+1+i{z}_{1,1}+i{z}_{1,2}+{z}_{1,1}{z}_{1,2}&+1-i{z}_{1,1}+i{z}_{1,2}-{z}_{1,1}{z}_{1,2}\\
		+1+i{z}_{1,1}-i{z}_{1,2}-{z}_{1,1}{z}_{1,2}&+1-i{z}_{1,1}-i{z}_{1,2}+{z}_{1,1}{z}_{1,2}
	\end{bmatrix}.
\end{equation}
Suppose $\bm{U}^{\{k\}}(\bm{z}_{k})$ ($k=0,2$) are PU matrix over QPSK, then
	\begin{equation}
		\begin{split}
			\bm{M}(\bm{z}_{0},\bm{z}_{1},\bm{z}_{2},{z}_{1},{z}_{2})
			&=\bm{U}^{\{0\}}(\bm{z}_{0})\cdot\bm{D}({z}_{1})\cdot\bm{U}^{\{1\}}(\bm{z}_{1})\cdot\bm{D}({z}_{2})\cdot\bm{U}^{\{3\}}(\bm{z}_{1,2}),\\
			&=\sum_{p=0}^{1}2^{1-p}\cdot\bm{U}^{\{0\}}(\bm{z}_{0})\cdot\bm{D}({z}_{1})\cdot\bm{U}^{\{1,p\}}(\bm{z}_{1})\cdot\bm{D}({z}_{2})\cdot\bm{U}^{\{3\}}(\bm{z}_{1,2}),
		\end{split}
	\end{equation}	
	is a PU matrix over $16$-QAM.
\end{example}

If $q$ is a composite integer, we can also construct PU matrices based on the factorization of $q$. The following is an example for the case $q=6$ and factorization $6=3\times2$. 
Denote the multiplication of sets by operator $\otimes$, i.e., ${S}_{1}\otimes{S}_{2}=\{\alpha\cdot\beta|\alpha\in{S}_{1},\beta\in{S}_{2}\}$.
Since  ${S}_{4^{6}\text{-QAM}} =\left\{{\bigoplus}_{p_1=0}^{2}2^{p_1}\cdot{S}_{\text{QPSK}}\right\}\otimes\left\{{\bigoplus}_{p_2=0}^{1}2^{3\cdot p_2}\cdot{S}_{\text{QPSK}}\right\}$, we have the following example.

\begin{example}
Suppose $\bm{U}^{\{0\}}(\bm{z}_{0})$ are PU matrix over QPSK and $\bm{U}^{\{1,p_1\}}(\bm{z}_{1})$ ($0\leq{p_1}<2$) and $\bm{U}^{\{2,p_2\}}(\bm{z}_{2})$ ($0\leq{p_2}<1$) are generating matrices over QPSK which satisfy that $\sum_{p_1=0}^{2}2^{p_1}\bm{{U}}^{\{1,{p_1}\}}(\bm{z}_{1})$ and $\sum_{p_2=0}^{1}2^{3\cdot p_2}\bm{U}^{\{1,{p_2}\}}(\bm{z}_{2})$ are PU matrix. Then
\begin{equation}
\begin{split}
\bm{M}(\bm{z},\bm{z}')
&=\bm{U}^{\{0\}}(\bm{z}_{0})\cdot\bm{D}({Z}_{1})\cdot\left(\sum_{p_1=0}^{2}2^{p_1}\bm{{U}}^{\{1,{p_1}\}}(\bm{z}_{1})\right)\cdot\bm{D}({Z}_{2})\cdot\left(\sum_{p_2=0}^{1}2^{3\cdot p_2}\bm{U}^{\{1,{p_2}\}}(\bm{z}_{2})\right)\\
&=\sum_{p_1=0}^{2}\sum_{p_2=0}^{1}2^{{p_1}+3\cdot p_2}\left(\bm{U}^{\{0\}}(\bm{z}_{0})\cdot\bm{D}({Z}_{1})\cdot\bm{{U}}^{\{1,{p_1}\}}(\bm{z}_{1})\cdot\bm{D}({Z}_{2})\cdot\bm{U}^{\{1,{p_2}\}}(\bm{z}_{2})\right),
\end{split}
\end{equation}	
is also a PU matrix over $4^{6}$-QAM.
\end{example}

If we want to identify the $p$-th generating matrices over QPSK, we only need to calculate $p_1$ and $p_2$ by the equation $p=3p_2+p_1$.
We now can see that mixed-radix representation plays an important role.

Let $q=q_0\times q_1\times\cdots\times q_n$ be an ordered factorization of $q$, where $q_{k}\ge1 \;(0\le k\le n)$ are positive integers. Then any integer $p$ $(0\le p\le q-1)$ can be represented by a mixed radix representation
$({p}_{0},{p}_{1},\dots,{p}_{n})$ based on radixes $({q}_{0},{q}_{1},\dots,{q}_{n})$ and permutation $\phi$, i.e.,
\begin{equation}\label{eq: p=[rho(p)]}
	p
=\sum_{k=0}^{n}{p}_{k}\cdot{Q}_{k},
\end{equation}
where ${Q}_{\phi(k)}=\prod_{i=0}^{k-1}{q}_{\phi(i)}$, $\phi$ is a permutation of $\{0,1,\dots,n\}$.	

Since  $\mathbf{S}_{4^{q}\text{-QAM}} ={\bigotimes}_{k=0}^{n}\left\{{\bigoplus}_{i_k=0}^{q_k-1}2^{(q_{k}-1-i_{k})\cdot{Q}_{k}}\cdot{S}_{\text{QPSK}}\right\}$, we have the following theorem.

\begin{theorem}[QAM Golay array matrix]\label{thm: QAM_CCA}
Suppose that $q_{k}$ and $Q_{k}$ ($0\leq{k}\leq{n}$) are given above.
In \eqref{eq: M(z)}, suppose that 
\begin{equation}\label{eq: QAM-U}
	\bm{{U}}^{\{k\}}(\bm{z}_{k})=\sum_{0\leq {i}_{k}<q_{k}}2^{(q_{k}-1-{i}_{k})\cdot{Q}_{k}}\cdot\bm{{U}}^{\{k,{p}_{k}\}}(\bm{z}_{k}),
\end{equation}
are PU matrix,
where for $ 0\leq{p}_{k}<{q}_{k} $, 
\begin{equation}\label{eq: array_U^{k,p}}
	\bm{{U}}^{\{k,{p}_{k}\}}(\bm{z}_{k})=\left[\begin{matrix}
		{F}_{0,0}^{\{k,{p}_{k}\}}({\bm{z}_{k}})&{F}_{0,1}^{\{k,{p}_{k}\}}({\bm{z}_{k}})\\
		{F}_{1,0}^{\{k,{p}_{k}\}}({\bm{z}_{k}})&{F}_{1,1}^{\{k,{p}_{k}\}}({\bm{z}_{k}})\\
	\end{matrix}\right]
\end{equation} 
are QPSK generating matrices whose corresponding function matrix is denoted by
\begin{equation}\label{eq: U^{k,p}(y_k)}
	\widetilde{\bm{U}}^{\{k,{p}_{k}\}}(\bm{y}_{k})=\left[\begin{matrix}
		{f}_{0,0}^{\{k,{p}_{k}\}}({\bm{y}_{k}})&{f}_{0,1}^{\{k,{p}_{k}\}}({\bm{y}_{k}})\\
		{f}_{1,0}^{\{k,{p}_{k}\}}({\bm{y}_{k}})&{f}_{1,1}^{\{k,{p}_{k}\}}({\bm{y}_{k}})\\
	\end{matrix}\right].
\end{equation} 
Then $\bm{M}(\bm{z},\bm{z}')$ 
	is a generating matrix of $4^{q}$-QAM Golay array matrix, whose function matrix  $\widetilde{\mathbb{M}}(\bm{x},\bm{y})=
		\left\{{f}_{i,j}^{(p)}(\bm{x},\bm{y})\middle|0\leq{i,j}\leq1;0\leq{p}<q\right\}$ is  given by
\begin{equation}\label{eq: f^(p)(x)}
	\begin{split}
		{f}^{(p)}_{i,j}(\bm{x},\bm{y})
		=&\sum_{k=0}^{n}
		\left({f}_{0,0}^{\{k,{p}_{k}\}}(\bm{y}_{k})\cdot\bar{x}_{k}\bar{x}_{k+1}
		+{f}_{1,0}^{\{k,{p}_{k}\}}(\bm{y}_{k})\cdot{x}_{k}\bar{x}_{k+1}\right.\\
		&\quad+\left.{f}_{0,1}^{\{k,{p}_{k}\}}(\bm{y}_{k})\cdot\bar{x}_{k}{x}_{k+1}
		+{f}_{1,1}^{\{k,{p}_{k}\}}(\bm{y}_{k})\cdot{x}_{k}{x}_{k+1}\right),
	\end{split}
\end{equation}	
where $x_{0}=i$ and $x_{n+1}=j$ for simplicity, and ${p}_{k}$ ($0\leq{k}\leq{n}$) are decided by \eqref{eq: p=[rho(p)]}.
	
\end{theorem}

\begin{proof}
Since $\bm{U}^{\{k\}}(\bm{z}_{k})$ are  PU matrices,
then $\bm{M}(\bm{z},\bm{z}')$ in \eqref{eq: M(z)} is a PU matrix over $4^{q}$-QAM constellation, which  can be represented by weighted sum of generating matrices over QPSK, i.e.,
\begin{equation}
\bm{M}(\bm{z},\bm{z}')
	=\sum_{p=0}^{q-1}2^{q-1-p}\bm{M}^{(p)}(\bm{z},\bm{z}')
\end{equation}
where 
\begin{equation}
	\begin{aligned}
		\bm{M}^{(p)}(\bm{z},\bm{z}')
		=\bm{U}^{\{0,{p}_{0}\}}\cdot\prod_{k=1}^{n}\left(\bm{D}({z}_{k})\cdot\bm{U}^{\{k,{p}_{k}\}}\right).
	\end{aligned}
\end{equation}
The corresponding function \eqref{eq: f^(p)(x)} can be obtained  in  the similar way as in Theorem \ref{thm: PSK_CCA}.
\end{proof}

\begin{remark}
One of the important things to fulfill this construction is to find $\bm{{U}}^{\{k\}}(\bm{z}_{k})$
satisfying \eqref{eq: QAM-U}, i.e., PU matrix over ${\bigoplus}_{i_k=0}^{q_k-1}2^{(q_{k}-1-i_{k})\cdot{Q}_{k}}\cdot{S}_{\text{QPSK}}$.
We will study this problem in the next section.
\end{remark}

In \cite{fullpaper}, $\bm{{U}}^{\{k\}}(\bm{z}_{k})$ are all chosen from Hadamard matrices.
While in this paper, we give a more generalized form, i.e., $\bm{{U}}^{\{k\}}(\bm{z}_{k})$ can be chosen from either Hadamard matrices or PU matrices. The details will be given in the next section.

\subsection{Evaluated Golay Sequence Matrix}

Since the constructed Golay array matrices in Theorems \ref{thm: PSK_CCA} and \ref{thm: QAM_CCA} are of size 
$\bm{2}^{(n)}\times\bm{b}_0\times \bm{b}_1\cdots\times \bm{b}_n$,
the corresponding affine function ${l}(\bm{x},\bm{y})$ can  be given by
\begin{equation}\label{eq: {l}(x,y)}
	{{l}(\bm{x},\bm{y})}
	=\sum_{k=1}^{n}{c}_{k}{x}_{k}+\sum_{k=0}^{n}\sum_{v=1}^{r_k}{c}_{k,v}{y}_{k,v}+{c}_{0},
\end{equation}
where ${c}_{k},{c}_{k,v}\in\Z_{H}$ for ${H}$-PSK case and ${c}_{k},{c}_{k,v}\in\Z_{4}$ for QAM case.
According to Property \ref{prop: QAM_linear},
array matrices described by
\begin{subequations}\label{eq: fun_M(x,y)}
\begin{alignat}{1}
\widetilde{\bm{M}}'(\bm{x},\bm{y})&=\widetilde{\bm{M}}(\bm{x},\bm{y})+{l}(\bm{x},\bm{y})\cdot\bm{J},\\
\widetilde{\mathbb{M}}'(\bm{x},\bm{y})&=\widetilde{\mathbb{M}}(\bm{x},\bm{y})+{l}(\bm{x},\bm{y})\cdot\vec{1}_{q}\cdot\bm{J},
\end{alignat}
\end{subequations}	
are Golay array matrices over ${H}$-PSK and $4^{q}$-QAM respectively (where ${l}(\bm{x},\bm{y})$ is given by \eqref{eq: {l}(x,y)}, $\widetilde{{\bm{M}}}(\bm{x},\bm{y})$ and $\widetilde{\mathbb{M}}({\bm{y}})$ are given by \eqref{eq: non-standard CCA} and \eqref{eq: f^(p)(x)} respectively).

The sequence matrices $\widetilde{\bm{M}}'(y)$ (resp. $\widetilde{\mathbb{M}}'(y)$) projected from $\widetilde{\bm{M}}'(\bm{x},\bm{y})$ (resp. $\widetilde{\mathbb{M}}'(\bm{x},\bm{y})$) form Golay sequence matrices over ${H}$-PSK (resp. $4^{q}$-QAM) by restricting $({x}_1,\dots,{x}_{n},{y}_{0,1},\dots,{y}_{0,r_1},\dots,{y}_{n,1},\dots{y}_{n,r_n})$ to be mixed radix representation of $y$ based on radixes 
$(2,\dots{2},{b}_{0,1},\dots,{b}_{0,r_1},\dots,{b}_{n,1},\dots,{b}_{n,r_n})$ and any permutation.



\section{Some Results}\label{Sec: QAM GCP and V-GBF}	

In this section we will introduce some known component PU matrices and give corresponding results based on Theorems \ref{thm: PSK_CCA} and \ref{thm: QAM_CCA}.
Since our constructions in last section are recursive ones.
We will specify the ``seed" PU matrices 
that can not be computed by the form $\bm{U}^{\{1\}}(\bm{z}_{1})\cdot\bm{D}({z})\cdot\bm{U}^{\{2\}}(\bm{z}_{2})$.

\subsection{Standard Golay Array matrices over ${H}$-PSK}

Denote $\mathcal{H}(N,{H})$ as {\em Butson-type} \cite{Butson62} Hadamard matrices of order $N$ with
entries being ${H}$-th roots of unity. $\mathcal{H}(2,{H})$ are most simple PU matrices (with no variables) and meanwhile their corresponding array matrices (with dimension $0$)
over QPSK. 
\begin{definition}[Butson-type Hadamard matrix]
can be  uniquely determined by $d_{0}, d_{1}, d_{2}\in \Z_{H}$, which can be denoted as
\begin{equation}
	\bm{{H}}({d_{0}}, {d_{1}}, {d_{2}})= \xi^{d_{0}}\cdot\begin{bmatrix}
		1&\\
		& \xi^{d_{1}}
	\end{bmatrix}\begin{bmatrix}
		1&1\\
		1&-1
	\end{bmatrix}\begin{bmatrix}
		1&\\
		& \xi^{d_{2}}
	\end{bmatrix}=\begin{bmatrix}
		\xi^{d_{0}}&\xi^{d_{0}+d_{2}}\\
		\xi^{d_{0}+d_{1}}&-\xi^{d_{0}+d_{1}+d_{2}}
	\end{bmatrix},
\end{equation}
\end{definition}
The corresponding function matrix is denoted by
\begin{equation}\label{eqn_H(ddd)}
	\widetilde{\bm{{H}}}({d_{0}}, {d_{1}}, {d_{2}})
	=\begin{bmatrix}
		{d_{0}}&{d_{0}+d_{2}}\\
		{d_{0}+d_{1}}&{d_{0}+d_{1}+d_{2}}+H/2
	\end{bmatrix}.
\end{equation}
%
In Theorem \ref{thm: PSK_CCA},
if $\bm{U}^{\{k\}}(\bm{z}_{k})$ are all chosen as $\bm{H}(0,0,0)$, i.e.,
${f}_{0,0}^{\{k\}}={f}_{0,1}^{\{k\}}={f}_{1,0}^{\{k\}}=0$, ${f}_{1,1}^{\{0\}}={H}/{2}$, we will have the following corollary.
\begin{corollary}[Standard GAM]\label{coro: st_CCA}
Array matrix described by
$\widetilde{{\bm{M}}}(\bm{x})=\left\{{f}_{i,j}(\bm{x})\middle|0\leq{i,j}\leq1\right\}$ where
\begin{equation}\label{eqn: st_GCS}
{f}_{i,j}(\bm{x})=\frac{H}{2}\cdot\sum_{k=1}^{n-1}x_{k}x_{k+1}+\frac{H}{2}i\cdot{x}_{1}+\frac{H}{2}j\cdot{x}_{n},
\end{equation}
is a Golay array matrix of size $\bm{2}^{(n)}$ over ${H}$-PSK.
\end{corollary}

Corollary \ref{coro: st_CCA} can be seen in \cite[Construction 1]{CCA}, the corresponding PU matrix is also given by \cite{Budi2018PU}.
The derived Golay sequence pairs are called standard Golay sequence pairs and the sequences projected from \eqref{eqn: st_GCS} are called standard Golay sequences \cite{Davis1999Peak}.

\subsection{``Cross-over" PU and Non-standard Golay array matrix over QPSK}
Golay sequences not of form \eqref{eqn: st_GCS} (including that adding affine items) were found by computer search \cite{Li05}, and were referred to non-standard Golay sequences.  
It was explained by cross-over Golay sequence pair of length $8$ \cite{Fiedler06} later on, which is present in the following fact.
 

\begin{fact}\cite[Theorem 12]{Array2}\cite{Fiedler06}
Each of the sequence pairs in the set
\[
\begin{split}
P:=\{&
(\mathcal{A}(y),\mathcal{B}(y)),
(\mathcal{A}(y),\mathcal{B}^{*}(y)),
(\mathcal{A}^{*}(y),\mathcal{B}(y)),
(\mathcal{A}^{*}(y),\mathcal{B}^{*}(y)),\\
&(\mathcal{B}(y),\mathcal{A}(y)),
(\mathcal{B}(y),\mathcal{A}^{*}(y)),
(\mathcal{B}^{*}(y),\mathcal{A}(y)),
(\mathcal{B}^{*}(y),\mathcal{A}^{*}(y))
\}
\end{split}
\]	
is a cross-over Golay sequence pair of length $8$ over QPSK, where the functions of $\mathcal{A}(y),\mathcal{B}(y)$ are given by
\[
(a(y)|0\le\tau\le7)=(0,0,0,2,0,0,2,0),
\quad
(b(y)|0\le\tau\le7)=(0,1,1,2,0,3,3,2).
\]
All $512$ ordered cross-over Golay sequence pairs of length $8$ over $\Z_{4}$ occur as affine offsets of the 8 pairs in $P$.
\end{fact}

\begin{definition}[Cross-over PU]\label{def: Cross-over PU}
According to Lemma \ref{lem: GAP->GAM}, define 
\begin{equation}\label{eq: fun_ABCD}
{\bm{S}}({y})=\begin{bmatrix}
\mathcal{S}_{0,0}(y)&\mathcal{S}_{0,1}(y)\\
\mathcal{S}_{1,0}(y)&\mathcal{S}_{1,1}(y)\\
\end{bmatrix},
\end{equation}
where $(\mathcal{S}_{0,0}(y),\mathcal{S}_{0,1}(y))\in P$, and $(\mathcal{S}_{1,0}(y),\mathcal{S}_{1,1}(y))=(\mathcal{S}_{0,1}^{*}(y),-\mathcal{S}_{0,0}^{*}(y))$.
The corresponding generating matrix is given by
\begin{equation}
	{\bm{S}}(z)=\begin{bmatrix}
		{S}_{0,0}(z)&{S}_{0,1}(z)\\
		{S}_{1,0}(z)&{S}_{1,1}(z)\\
	\end{bmatrix}.
\end{equation}
The corresponding function matrix is given by
\begin{equation}\label{eq: fun_abcd}
	\widetilde{{\bm{S}}}({y})=\begin{bmatrix}
		{s}_{0,0}(y)&{s}_{0,1}(y)\\
		{s}_{1,0}(y)&{s}_{1,1}(y)\\
	\end{bmatrix}.
\end{equation}	
${\bm{S}}({y})$ is called cross-over sequence matrices and
${\bm{S}}(z)$ is called cross-over PU matrix.
\end{definition}
 

In this subsection, we chose $r$ out of $n+1$ of the seed PU in Theorem \ref{thm: PSK_CCA} as ``cross-over" PU matrix.
Thus the constructed arrays are of size 
$\bm{2}^{(n)}\times\bm{8}^{(r)}={2\times2\times\cdots\times2}\times{8\times 8\cdots\times8}$. 

For $r\leq{n}$, denote $\mathcal{K}=\left\{k_{1},k_{2},\dots,k_{r}\right\}\subseteq\{0,1,2,\dots,n\}$.
For convenience, set the Boolean variable $x_{k}\in\mathbf{Z}_{2}$ ($1\leq{k}\leq{n}$)  and $y_{k}\in\mathbf{Z}_{8}$ ($k\in\mathcal{K}$) whose corresponding polynomial variables are $z_{k}$ and $z_{k}'$ respectively, i.e.,
\begin{equation*}
	\left\{
	\begin{aligned}
		&(\bm{z},\bm{z}')=({z}_1,{z}_2,\dots,{z}_{n},{z}'_1,\dots,z'_{r});\\
		&(\bm{x},\bm{y})=({x}_1,{x}_2,\dots,{x}_{n},{y}_{k_{1}},\dots,{y}_{k_r})\in\mathbf{Z}_{2}^{n}\times \mathbf{Z}_8^{r}.\\
	\end{aligned}\right.
\end{equation*}

Suppose $\bm{U}^{\{k\}}(\bm{y}_{k})$ in \eqref{eq: M(z)} are chosen as cross-over PU matrices $\bm{S}^{\{k\}}(z_{k}')$ for $k\in\mathcal{K}$ and $\bm{H}(0,0,0)$ for $k\notin\mathcal{K}$, i.e., 
$\widetilde{\bm{U}}^{\{k\}}(\bm{y}_{k})$ in Theorem \ref{thm: PSK_CCA} are given by 
\begin{equation}\label{eq: function_U^{k}(y)}
	\widetilde{\bm{U}}^{\{k\}}(\bm{y}_{k})=\begin{bmatrix}
		{f}_{0,0}^{\{k\}}(\bm{y}_{k})&{f}_{0,1}^{\{k\}}(\bm{y}_{k})\\
		{f}_{1,0}^{\{k\}}(\bm{y}_{k})&{f}_{1,1}^{\{k\}}(\bm{y}_{k})
	\end{bmatrix}=\left\{\begin{aligned}
		&\widetilde{\bm{S}}^{\{k\}}(y_{k})&=&\left[\begin{matrix}
			{s}_{0,0}^{\{k\}}(y_{k})&{s}_{0,1}^{\{k\}}(y_{k})\\
			{s}_{1,0}^{\{k\}}(y_{k})&{s}_{1,1}^{\{k\}}(y_{k})
		\end{matrix}\right],&k\in\mathcal{K},\\
		&\widetilde{\bm{H}}(0,0,0)&=&\left[\begin{matrix}
			0&0\\
			0&{2}
		\end{matrix}\right],
		&k\notin\mathcal{K},\end{aligned}\right.
\end{equation}
where $\widetilde{\bm{S}}^{\{k\}}(y_{k})=\left[\begin{smallmatrix}{s}_{0,0}^{\{k\}}(y_{k})&{s}_{0,1}^{\{k\}}(y_{k})\\{s}_{1,0}^{\{k\}}(y_{k})&{s}_{1,1}^{\{k\}}(y_{k})\end{smallmatrix}\right]$ are function matrices of cross-over Golay sequence matrices.
Based on Theorem \ref{thm: PSK_CCA}, we get the following corollary.
\begin{corollary}[Non-standard Golay array matrix over QPSK]\cite{Array2}\label{Coro: NS_GCA_QPSK}
The array matrix described by $\widetilde{{\bm{M}}}(\bm{x},\bm{y})=\left\{{f}_{i,j}(\bm{x},\bm{y})\middle|0\leq{i,j}\leq1\right\}$ where
\begin{equation}\label{eq: f(x,y)_non-standard}
\begin{split}
f_{i,j}(\bm{x},\bm{y})
&=\sum_{k\in\mathcal{K}}\left({s}_{0,0}^{\{k\}}(y_{k})\cdot\bar{x}_{k}\bar{x}_{k+1}
+{s}_{0,0}^{\{k\}}(y_{k})\cdot{x}_{k}\bar{x}_{k+1}
+{s}_{0,0}^{\{k\}}(y_{k})\cdot\bar{x}_{k}{x}_{k+1}
+{s}_{0,0}^{\{k\}}(y_{k})\cdot{x}_{k}{x}_{k+1}\right)\\
&+2\cdot\sum_{k\notin\mathcal{K}}x_{k}x_{k+1},
\end{split}	
\end{equation}
is a Golay array matrix, where $x_{0}=i$ and $x_{n+1}=j$ for simplicity.
\end{corollary}
\begin{proof}
\eqref{eq: f(x,y)_non-standard} can be obtained straightforwardly by substituting the value of ${f}_{i,j}^{\{k\}}$ in \eqref{eq: non-standard CCA}
\end{proof}

\subsection{Compatible (Para-)Unitary matrices}

One of the bases to fulfill the construction of Golay array matrix over $4^{q}$-QAM in Theorem \ref{thm: QAM_CCA} is to find $\bm{{U}}^{\{k\}}(\bm{z}_{k})$
satisfying \eqref{eq: QAM-U}, i.e., PU matrix over ${\bigoplus}_{i_k=0}^{q_k-1}2^{(q_{k}-1-i_{k})\cdot{Q}_{k}}\cdot{S}_{\text{QPSK}}$. 
A direct idea is to use the set of PU matrices, whose arbitrarily weighted sum are still PU matrices, i.e., \emph{compatible} PU matrices defined as follows.
\begin{definition}[Compatible (Para-)Unitary matrices]
The set of (para-)unitary matrices $\bm{U}^{\{k\}}(\bm{z})$ $(1\leq{k}<K)$ are called \emph{compatible} if for any $c_{k}\in\R$ $(1\leq{k}<K)$,
\begin{equation}
\bm{U}(\bm{z})=\sum_{k=1}^{K}c_{k}\bm{U}^{\{k\}}(\bm{z})
\end{equation}
is still a (para-)unitary matrix.
\end{definition}


We will introduce two examples of compatible (para-)unitary matrices.

For the unitary matrix case,
an example 
is based on the following fact:
The compatible  unitary matrix of the form $\bm{U}=\left[\begin{smallmatrix}
\alpha&{\beta}\\
\bar{\beta}&-\bar{\alpha}
\end{smallmatrix}\right]$, where $\alpha,\beta\in\C$, is closed under matrix addition.
The unitary matrix in this form over QPSK can be presented by compatible Hadamard matrices as follows.
\begin{lemma}[Compatible Hadamard Matrices]\label{lem: Compatible H}
For ${d_{0}}, {d_{1}}, {d_{2}}\in\Z_{4} $ such that $2d_{0}+d_{1}+d_{2}=0$, ${\bm{{H}}}({d_{0}}, {d_{1}}, {d_{2}})\in H(2,4)$ are compatible Hadamard matrices. The corresponding function matrix is $\widetilde{\bm{{H}}}({d_{0}}, {d_{1}}, {d_{2}})$.
\end{lemma}
\begin{proof}
For any weights $c_{d_{0},d_{1},d_{2}}\in\R$, it's easy to verify
\begin{equation}
\mathbb{U}=\sum
c_{d_{0},d_{1},d_{2}}\cdot{\bm{{H}}}({d_{0}}, {d_{1}}, {d_{2}})=\left[\begin{smallmatrix}
	\alpha&{\beta}\\
	\bar{\beta}&-\bar{\alpha}
\end{smallmatrix}\right]
\end{equation}
where $\alpha=\sum c_{d_{0},d_{1},d_{2}}\cdot\xi^{d_{0}}$ and $\beta=\sum c_{d_{0},d_{1},d_{2}}\cdot\xi^{d_{0}+d_{2}}$,
is a unitary matrix.
\end{proof}
\begin{remark}
Of course, the unitary matrix of the form $\bm{U}=\left[\begin{smallmatrix}
{1}&0\\
0&\xi^{c}
\end{smallmatrix}\right]\cdot\left[\begin{smallmatrix}
	\alpha&{\beta}\\
	\bar{\beta}&-\bar{\alpha}
\end{smallmatrix}\right]$, where $\alpha,\beta\in\C$, $c\in\Z_{4}$, is also closed under matrix addition.
So 
Lemma \ref{lem: Compatible H} is also valid if we restrict $2d_{0}+d_{1}+d_{2}=c$. However, we only need to study the case $2d_{0}+d_{1}+d_{2}=0$, since they are equivalent in terms of non-linear part from the viewpoint of the  of the derived functions. 
\end{remark}

Based on Lemma \ref{lem: Compatible H}, we can construct unitary matrix over desired constellation.
For example,
\begin{equation}
	4\cdot\bm{H}(0,0,0)+2\cdot\bm{H}(2,1,3)+\bm{H}(1,1,1)=4\cdot\begin{bmatrix}
	1&1\\
	1&-1
	\end{bmatrix}+2\cdot\begin{bmatrix}
	-1&i\\
	-i&1
	\end{bmatrix}+1\cdot\begin{bmatrix}
	i&-1\\
	-1&i
	\end{bmatrix},
\end{equation}
is a unitary matrix over $4\cdot{S}_{\text{QPSK}}\oplus2\cdot{S}_{\text{QPSK}}\oplus1\cdot{S}_{\text{QPSK}}={S}_{4^{3}\text{-QAM}}$.

For the PU matrix case,
an example is based on the following fact: The product of a PU matrix and a unitary matrix is still a PU matrix.
For $a,b,c,d\in\R$,
$\bm{U}=\left[\begin{smallmatrix}
	a+b\cdot\xi&c+d\cdot\xi\\
	-c+d\cdot\xi&a-b\cdot\xi
\end{smallmatrix}\right]$ is a unitary matrix.
Suppose $\bm{U}(\bm{z})$ is PU matrix.
Then $\bm{U}(\bm{z})\cdot\bm{U}$ is still a PU matrix.
On the other hand, 
$\bm{U}=a\cdot\bm{I}+b\cdot\bm{i}+c\cdot\bm{j}+d\cdot\bm{k}$ where
\[\bm{I}=\begin{bmatrix}
	1&0\\
	0&1
\end{bmatrix},\quad
\bm{i}=\begin{bmatrix}
	\xi&0\\
	0&-\xi
\end{bmatrix},\quad
\bm{j}=\begin{bmatrix}
	0&1\\
	-1&0
\end{bmatrix},\quad
\bm{k}=\begin{bmatrix}
	0&\xi\\
	\xi&0
\end{bmatrix}.\]
So that $a\cdot\bm{U}(\bm{z})\cdot\bm{I}+b\cdot\bm{U}(\bm{z})\cdot\bm{i}+c\cdot\bm{U}(\bm{z})\cdot\bm{j}+d\cdot\bm{U}(\bm{z})\cdot\bm{k}=\bm{U}(\bm{z})\cdot\bm{U}$ is a PU matrix, which means
 $\bm{U}(\bm{z})\cdot\bm{I}$,  $\bm{U}(\bm{z})\cdot\bm{i}$, $\bm{U}(\bm{z})\cdot\bm{j}$,  $\bm{U}(\bm{z})\cdot\bm{k}$ are compatible since $a,b,c,d$ are arbitrary.
Similarly, $\bm{I}\cdot\bm{U}(\bm{z})$,  $\bm{i}\cdot\bm{U}(\bm{z})$, $\bm{j}\cdot\bm{U}(\bm{z})$,  $\bm{k}\cdot\bm{U}(\bm{z})$ are also a set of compatible PU matrices.

The set of matrices $\{\bm{I},\bm{i},\bm{j},\bm{k},-\bm{I},-\bm{i},-\bm{j},-\bm{k}\}$ are matrix form of  \emph{Unit Quaternions},
which can be uniformly expressed by $
\bm{i}^{e}\cdot\bm{j}^{v}$,
where $e\in\Z_{4}$, $v\in\mathbf{Z}_{2}$, i.e.,
\[\begin{array}{llll}
 \bm{I}=\bm{i}^{0}\cdot\bm{j}^{0},&
 \bm{i}=\bm{i}^{1}\cdot\bm{j}^{0},&
-\bm{I}=\bm{i}^{2}\cdot\bm{j}^{0},&
-\bm{i}=\bm{i}^{3}\cdot\bm{j}^{0},\\
 \bm{j}=\bm{i}^{0}\cdot\bm{j}^{1},&
 \bm{k}=\bm{i}^{1}\cdot\bm{j}^{1},&
-\bm{j}=\bm{i}^{2}\cdot\bm{j}^{1},&
-\bm{k}=\bm{i}^{3}\cdot\bm{j}^{1}.
\end{array}\]

The compatible PU matrices based on matrix form of Unit Quaternions over QPSK can be presented as follows.
\begin{lemma}[Compatible PU Matrices]\label{lem: Compatible PU}
Suppose $\bm{U}(\bm{z})$ is PU matrix over $S_{1}$, and the $(i,j)$-th entry in the corresponding function matrix $\widetilde{\bm{U}}(\bm{y})$ is ${f}_{i,j}(\bm{y})$.
Then  $\bm{U}_{L}[v,e](\bm{z})=\bm{i}^{e}\cdot\bm{j}^{v}\cdot\bm{U}(\bm{z})$ and $\bm{U}_{R}[v,e](\bm{z})=\bm{U}(\bm{z})\cdot(-\bm{j})^{v}\cdot\bm{i}^{e}$ ($e\in\Z_{4}$, $v\in\{0,1\}$)
are two sets of compatible PU matrices over $S_{1}$, whose corresponding function matrices are
\begin{equation}\label{eq: fun_U_L[v,e]}
\begin{split}
 \widetilde{\bm{U}}_{L}[v,e](\bm{y})&=
		\begin{bmatrix}
			{f}_{v,0}({\bm{y}})&{f}_{v,1}({\bm{y}})\\
			{f}_{\bar{v},0}(\bm{y})&{f}_{\bar{v},1}(\bm{y})\\
		\end{bmatrix}+v\begin{bmatrix}
			0&0\\
			2&2
		\end{bmatrix}+e\begin{bmatrix}
			1& 1\\
			-1&-1
		\end{bmatrix},
	\end{split}
\end{equation} 	
\begin{equation}\label{eq: fun_U_R[v,e]}
\widetilde{\bm{U}}_{R}[v,e](\bm{y})=
\begin{bmatrix}
{f}_{0,{v}}({\bm{y}})&{f}_{0,\bar{v}}({\bm{y}})\\
{f}_{1,{v}}(\bm{y})&{f}_{1,\bar{v}}(\bm{y})\\
\end{bmatrix}+v\cdot\begin{bmatrix}
0&2\\
0&2
\end{bmatrix}+e\cdot\begin{bmatrix}
1&-1\\
1&-1
\end{bmatrix},
\end{equation} 	 		
where $\bar{v}=1-v$.
\end{lemma}
\begin{proof}
For arbitrary  $c_{v,e}\in\R$ ($e\in\Z_{4},v\in\{0,1\}$), it's to verify
\begin{equation}
	\mathbb{U}(z)=\sum_{e\in\Z_{4},v\in\{0,1\}}c_{v,e}\cdot\bm{U}_{L}[v,e](\bm{z})
	=\left[\begin{smallmatrix}
		\alpha&{\beta}\\
		\bar{\beta}&-\bar{\alpha}
	\end{smallmatrix}\right]\cdot\bm{U}(\bm{z}),
\end{equation}
where $\alpha=\sum_{e\in\Z_{4}}c_{0,e}\cdot\xi^{e}$ and $\beta=\sum_{e\in\Z_{4}}c_{1,e}\cdot\xi^{e}$,
is a PU matrix.
The function matrix $\widetilde{\bm{U}}[v,e](\bm{y})$ can be obtained easily, since  
\begin{equation}
\bm{U}_{L}[0,e](\bm{z})=\left[\begin{matrix}
F_{0,0}(\bm{z})\xi^{e}&F_{0,1}(\bm{z})\xi^{e}\\
F_{1,0}(\bm{z})\xi^{-e}&F_{1,1}(\bm{z})\xi^{-e}
\end{matrix}\right], 
\end{equation}
\begin{equation}
\bm{U}_{L}[1,e](\bm{z})=\left[\begin{matrix}
F_{1,0}(\bm{z})\xi^{e}&-F_{1,1}(\bm{z})\xi^{e}\\
F_{0,0}(\bm{z})\xi^{-e}&-F_{0,1}(\bm{z})\xi^{-e}
\end{matrix}\right].
\end{equation}
The compatibility of $\bm{U}_{R}[v,e](\bm{z})$ and the expression of $\widetilde{\bm{U}}_{R}[v,e](\bm{z})$ can be proved similarly.
\end{proof}	

\begin{remark}
If $\bm{U}(\bm{z})$ is a Hadamard matrix, say $\bm{U}(\bm{z})=\bm{H}(d_{0},d_{1},d_{2})$,
it is easy to verify that $\bm{i}^{e}\cdot\bm{j}^{v}\cdot\bm{H}(d_{0},d_{1},d_{2})=\bm{H}(d_{0}',d_{1}',d_{2}')$, where $d_{0}'=d_{0}+vd_{1}+e$, $d_{1}'=(1+2v)d_{1}+2v+2e$ and $d_{2}'=d_{2}+2v$. Notice that $2d_{0}+d_{1}+d_{2}=2d_{0}'+d_{1}'+d_{2}'$ (mod $4$). It means that the 
$\{\bm{i}^{e}\cdot\bm{j}^{v}\cdot\bm{H}(d_{0},d_{1},d_{2})|e\in\Z_{4},v\in\mathbf{Z}_{2}\}$ is a subset of compatible Hadamard matrix given by Lemma \ref{lem: Compatible H}.
\end{remark}

Based on Lemma \ref{lem: Compatible PU}, we can construct PU matrix over desired constellation from PU matrix over QPSK.
For example, suppose $\bm{U}(\bm{z})$ is a PU matrix over QPSK, then $\sum_{p=0}^{q-1}2^{p}\cdot\bm{U}[e_{p},v_{p}](\bm{z})$ is a PU matrix over $4^{q}$-QAM.
A more specific example is that, suppose $\bm{S}(z)$ is a cross-over PU matrix, then
\begin{equation}
	\begin{split}
		&4\cdot\bm{S}({z})\cdot\bm{I}+2\cdot\bm{S}({z})\cdot\bm{k}+\bm{S}({z})\cdot\bm{i}\\&=4\cdot\begin{bmatrix}
			S_{0,0}(z)&S_{0,1}(z)\\
			S_{1,0}(z)&S_{1,1}(z)
		\end{bmatrix}+2\cdot\xi\cdot\begin{bmatrix}
			S_{0,1}(z)&S_{0,0}(z)\\
			S_{1,1}(z)&S_{1,0}(z)
		\end{bmatrix}+\xi\cdot \begin{bmatrix}
			S_{0,0}(z)&-S_{0,1}(z)\\
			S_{1,0}(z)&-S_{1,1}(z)
		\end{bmatrix},  
	\end{split}
\end{equation}
is a PU matrix over $4\cdot{S}_{\text{QPSK}}\oplus2\cdot{S}_{\text{QPSK}}\oplus1\cdot{S}_{\text{QPSK}}={S}_{4^{q}\text{-QAM}}$, 
which is never reported before.

\subsection{Non-standard Golay Array Matrices}
For any PU matrix over $S_{1}$ which is not seed PU matrix, say, $\bm{U}(\bm{z}_{1},\bm{z}_{2},z)=\bm{U}^{\{1\}}(\bm{z}_{1})\cdot\bm{D}(z)\cdot\bm{U}^{\{2\}}(\bm{z}_{2})$, we have
$\sum_{p=0}^{q-1}2^{p}\cdot\bm{U}_{L}[e_{p},v_{p}](\bm{z}_{1},\bm{z}_{2},z)=\left(\sum_{p=0}^{q-1}2^{p}\cdot\bm{U}_{L}^{\{1\}}[e_{p},v_{p}](\bm{z}_{1})\right)\cdot\bm{D}(z)\cdot\bm{U}^{\{2\}}(\bm{z}_{2})$.
Thus to construct PU matrix over $4^{q}$-QAM,
we only need to study the compatible seed PU matrices.
So in our construction, we will chose $\widetilde{\bm{U}}^{\{k\}}(\bm{z}_{k})$ to be either compatible ``cross-over" PU matrices based on Lemma \ref{lem: Compatible PU} or compatible Hadamard matrices based on Lemma \ref{lem: Compatible H}.


For $r\leq{n}$, denote $\mathcal{K}=\left\{k_{1},k_{2},\dots,k_{r}\right\}\subseteq\{0,1,2,\dots,n\}$.
For convenience, set the Boolean variable $x_{k}\in\mathbf{Z}_{2}$ ($1\leq{k}\leq{n}$)  and $y_{k}\in\mathbf{Z}_{8}$ ($k\in\mathcal{K}$) whose corresponding polynomial variables are $z_{k}$ and $z_{k}'$ respectively, i.e.,
\begin{equation*}
\left\{
\begin{aligned}
&(\bm{z},\bm{z}')=({z}_1,{z}_2,\dots,{z}_{n},{z}'_1,\dots,z'_{r});\\
&(\bm{x},\bm{y})=({x}_1,{x}_2,\dots,{x}_{n},{y}_{k_{1}},\dots,{y}_{k_r})\in\mathbf{Z}_{2}^{n}\times \mathbf{Z}_8^{r}.\\
\end{aligned}\right.
\end{equation*}
In Theorem \ref{thm: QAM_CCA}, if we choose $\bm{{U}}_{k}^{({p}_{k})}(\bm{z}_{k})$ ($0\leq{k}\leq{n}, 0\leq{p_{k}}<{q_{k}}$) in \eqref{eq: array_U^{k,p}} to be 
\begin{equation}\label{eq: U^{k}}
	\bm{U}^{\{k,{p}_{k}\}}(\bm{z}_{k})=\left\{
	\begin{aligned}
		&\bm{S}_{X}^{\{k\}}\left(v^{\{k,{p}_{k}\}},e^{\{k,{p}_{k}\}}\right)({z}_{k}),
		&(k\in\mathcal{K});\\
		&\bm{H}\left(d_{0}^{\{k,{p}_{k}\}},d_{1}^{\{k,{p}_{k}\}},d_{2}^{\{k,{p}_{k}\}}\right), &(k\notin\mathcal{K}).
	\end{aligned}\right.
\end{equation}
where for $ k\in\mathcal{K}$, the symbol $X\in\{L,R\}$, $\bm{S}_{L}^{\{k\}}\left(v,e\right)({z}_{k})=\bm{i}^{e}\cdot\bm{j}^{v}\cdot\bm{S}^{\{k\}}({z}_{k})$, $
\bm{S}_{R}^{\{k\}}\left(v,e\right)({z}_{k})=\bm{S}^{\{k\}}({z}_{k})\cdot(-\bm{j})^{v}\cdot\bm{i}^{e}$,
and $\bm{S}^{\{k\}}({z}_{k})$ are cross-over PU introduced in Definition \ref{def: Cross-over PU}, $e^{\{k,{p}_{k}\}}\in\Z_{4}$, ${v}^{\{k,{p}_{k}\}}\in\mathbf{Z}_{2}$;
for $ k\notin\mathcal{K}$, 
$\left({d_{0}^{\{k,{p}_{k}\}}}, {d_{1}^{\{k,{p}_{k}\}}}, {d_{2}^{\{k,{p}_{k}\}}}\right)\in\Z_{4}^{3}$ which satisfy
$2{d_{0}^{\{k,{p}_{k}\}}}+{d_{1}^{\{k,{p}_{k}\}}}+{d_{2}^{\{k,{p}_{k}\}}}=0$.
Based on Theorem \ref{thm: QAM_CCA}, we get the following corollary.
\begin{corollary}[Non-standard Golay array matrix over QAM]\label{Coro: NS_GCA_QAM}
The array matrix described by 
$\widetilde{\mathbb{M}}(\bm{x},\bm{y})=
\left\{{f}_{i,j}^{(p)}(\bm{x},\bm{y})\middle|0\leq{i,j}\leq1;0\leq{p}<q\right\}$
is a Golay array matrix over $4^{q}$-QAM, if
\begin{equation}\label{eq: f(x,y)_QAM_NS_GAM}
{f}_{i,j}^{(p)}(\bm{x},\bm{y})=\sum_{k=0}^{n}{f}^{\{k,{p}_{k}\}},
\end{equation}
where
\begin{equation}
{f}^{\{k,{p}_{k}\}}
=\left\{
\begin{aligned}
&{f}_{X}^{\{k\}}[{v}^{\{k,{p}_{k}\}},{e}^{\{k,{p}_{k}\}}]({x}_{k},{x}_{k+1},{y}_{k}),&(k\in\mathcal{K});\\
&{f}[d_{0}^{\{k,{p}_{k}\}},d_{1}^{\{k,{p}_{k}\}},d_{2}^{\{k,{p}_{k}\}}](x_{k},x_{k+1}), &(k\notin\mathcal{K}),
\end{aligned}\right.
\end{equation}
where 
\begin{equation}
{f}[d_{0},d_{1},d_{2}](x_{1},x_{2})=2x_{1}x_{2}+d_{0}+d_{1}x_{1}+d_{2}x_{2},
\end{equation}
and for $k\in\mathcal{K}$,
\begin{equation}
\begin{split}
&{f}_{L}^{\{k,p_{k}\}}[v,e]({x}_{1},{x}_{2},{y})\\
&={s}_{v,0}^{\{k\}}({y})\cdot\bar{x}_{1}\bar{x}_{2}
+{s}_{v,1}^{\{k\}}({y})\cdot{x}_{1}\bar{x}_{2}+{s}_{\bar{v},0}^{\{k\}}({y})\cdot\bar{x}_{1}{x}_{2}
+{s}_{\bar{v},1}^{\{k\}}({y})\cdot{x}_{1}{x}_{2}\\
&+(2v-2e)\cdot{x}_{1}+e,
\end{split}
\end{equation}
\begin{equation}
\begin{split}
&{f}_{R}^{\{k,p_{k}\}}[v,e]({x}_{1},{x}_{2},{y})\\
&={s}_{0,v}^{\{k\}}({y})\cdot\bar{x}_{1}\bar{x}_{2}
+{s}_{0,\bar{v}}^{\{k\}}({y})\cdot{x}_{1}\bar{x}_{2}+{s}_{1,{v}}^{\{k\}}({y})\cdot\bar{x}_{1}{x}_{2}
+{s}_{1,\bar{v}}^{\{k\}}({y})\cdot{x}_{1}{x}_{2}\\
&+(2v-2e)\cdot{x}_{2}+e, 
\end{split}
\end{equation}
where ${p}_{k}$ is decided by \eqref{eq: p=[rho(p)]}, where $x_{0}=i$ and $x_{n+1}=j$ for simplicity.
\end{corollary}

\begin{proof}	
Similar to the proof of Corollary \ref{Coro: NS_GCA_QPSK}, \eqref{eq: f(x,y)_QAM_NS_GAM} can be obtained by substituting $\widetilde{\bm{U}}^{\{k,p_{k}\}}(\bm{y}_{k})=\widetilde{\bm{S}}^{\{k\}}[{v}^{\{k,{p}_{k}\}},{e}^{\{k,{p}_{k}\}}](y_{k})$ (in the form of \eqref{eq: fun_U_L[v,e]}) or $\widetilde{\bm{S}}'^{\{k\}}[{v}^{\{k,{p}_{k}\}},{e}^{\{k,{p}_{k}\}}](y_{k})$ (in the form of \eqref{eq: fun_U_R[v,e]}) for $k\in\mathcal{K}$ and $\widetilde{\bm{H}}\left({d_{0}^{\{k,{p}_{k}\}}}, {d_{1}^{\{k,{p}_{k}\}}}, {d_{2}^{\{k,{p}_{k}\}}}\right)$ (in the form of \eqref{eqn_H(ddd)}) for $k\notin\mathcal{K}$ to \eqref{eq: f^(p)(x)}.
	
\end{proof}
\begin{remark}
If $ {r}=0 $, the result is equivalent to Theorem 1 in \cite{fullpaper}.
\end{remark}

\section{A Lower Bound of Enumeration}\label{Sec: Enumerations}
We will give the enumeration based on the ``three stage process''.

\subsection{Enumeration of Non-standard Golay array matrix over QAM}
In Corollary \ref{Coro: NS_GCA_QAM},
if we set 
$e^{\{k,0\}}=0$, $v^{\{k,0\}}=0$ ($ k\in\mathcal{K}$),
$\left({d_{0}^{\{k,0\}}}, {d_{1}^{\{k,0\}}}, {d_{2}^{\{k,0\}}}\right)\in(0,0,0)$
($ k\notin\mathcal{K}$), then $f_{i,j}^{(0)}(\bm{x},\bm{y})$ is same as \eqref{eq: f(x,y)_non-standard}.
Denote $\mu_{i,j}^{(p)}(\bm{x},\bm{y})=f_{i,j}^{(p)}(\bm{x},\bm{y})-f_{i,j}^{(0)}(\bm{x},\bm{y})$ as the offset, then,
\begin{equation}
\mu_{i,j}^{(p)}(\bm{x},\bm{y})=\sum_{k=0}^{n}\mu^{\{k,{p}_{k}\}},
\end{equation}
where
\begin{equation}
	\mu^{\{k,{p}_{k}\}}
	=\left\{
	\begin{aligned}
		&{\mu}_{X}^{\{k\}}[{v}^{\{k,{p}_{k}\}},{e}^{\{k,{p}_{k}\}}]({x}_{k},{x}_{k+1},{y}_{k}),&(k\in\mathcal{K});\\
		&{\mu}[d_{0}^{\{k,{p}_{k}\}},d_{1}^{\{k,{p}_{k}\}},d_{2}^{\{k,{p}_{k}\}}](x_{k},x_{k+1}), &(k\notin\mathcal{K}),
	\end{aligned}\right.
\end{equation}
where 
\begin{equation}
{\mu}[d_{0},d_{1},d_{2}](x_{1},x_{2})=d_{0}+d_{1}x_{1}+d_{2}x_{2},
\end{equation}
and
\begin{equation}
	\begin{split}
		{\mu}_{L}^{\{k\}}[v,e]({x}_{1},{x}_{2},{y})
		&={f}_{L}^{\{k\}}[v,e]({x}_{1},{x}_{2},{y})-{f}_{L}^{\{k\}}[0,0]({x}_{1},{x}_{2},{y})\\
		&=\left(\left({s}_{1,0}^{\{k\}}({y})-{s}_{0,0}^{\{k\}}({y})\right)\cdot\bar{x}_{1}+\left({s}_{1,1}^{\{k\}}({y})-{s}_{0,1}^{\{k\}}({y})\right)\cdot{x}_{1}\right)(1-2x_{2})\cdot{v}
		\\
		&+(2v-2e)\cdot{x}_{1}+e,
\end{split}
\end{equation}
\begin{equation}
\begin{split}
{\mu}_{R}^{\{k,p_{k}\}}[v,e]({x}_{1},{x}_{2},{y})&={f}_{R}^{\{k\}}[v,e]({x}_{1},{x}_{2},{y})-{f}_{R}^{\{k\}}[0,0]({x}_{1},{x}_{2},{y})\\&=\left(\left({s}_{0,1}^{\{k\}}({y})-{s}_{0,0}^{\{k\}}({y})\right)\cdot\bar{x}_{2}+\left({s}_{1,1}^{\{k\}}({y})-{s}_{1,0}^{\{k\}}({y})\right)\cdot{x}_{2}\right)(1-2x_{1})\cdot{v}\\
&+(2v-2e)\cdot{x}_{2}+e.\end{split}
\end{equation}

The $4^{q}$-QAM Golay array matrix is decided by $f_{i,j}^{(0)}(\bm{x},\bm{y})$ and $\vec{\mu}(\bm{x},\bm{y})=\left({\mu}^{(p)}(\bm{x},\bm{y})\middle|0\leq{p}<q\right)$.

Similar to \cite{Array2}, for fixed $(i,j)$, the number of $f_{i,j}^{(0)}(\bm{x},\bm{y})$ (refer to \eqref{eq: f(x,y)_non-standard}) is  determined by the choice of indices $\mathcal{K}$ of the $r=|\mathcal{K}|$ cross-over matrices, and by which of the $8$ possible values in $P$ each cross-over matrix takes. Therefore the number of allowed values for this sequence of pairs is 
\begin{equation}
\#\{f_{i,j}^{(0)}(\bm{x},\bm{y})\}=\binom{n+1}{r}8^{r}.
\end{equation}

Denote $\vec{v}^{\{k,{i}_{k}\}}$ ($(0\le k\le n,0\le{i}<q_{k})$) as $q$-dimensional vector consist of $0$ and $1$,
where the $p$-th entry equals $1$ if $p_{k}={i}_{k}$ and equals $0$ otherwise.
For example, suppose $n=1$, $q=6=2\times3$, then 
\[
\begin{array}{ccc}
	\vec{v}^{\{0,0\}}=(1,0,1,0,1,0),& \vec{v}^{\{0,1\}}=(0,1,0,1,0,1),&\\
	\vec{v}^{\{1,0\}}=(1,1,0,0,0,0),&
	\vec{v}^{\{1,1\}}=(0,0,1,1,0,0),&
	\vec{v}^{\{1,2\}}=(0,0,0,0,1,1).
\end{array}
\]
The subset $B={\{\vec{v}^{\{k,{i}_{k}\}}|0\le k\le n,0<{i}_{k}<q_{k}\}}$  (note $\vec{v}^{\{k,0\}}$ $(0\le k\le n)$ are excluded) are linearly independent.
(This can be easily verified. For example, the ${i}_{k}{Q}_{k}$-th entry of every vectors in $B$ is $0$ except for $\vec{v}^{\{k,{i}_{k}\}}$, which means $\vec{v}^{\{k,{i}_{k}\}}$ can not be given by linear summations of other vectors in $B$.)

The vectorial offset  $\vec{\mu}(\bm{x},\bm{y})$ can be viewed as a $(\sum_{k=0}^{n}{q}_{k}-n)$-dimensional function space spanned by $B$ over ${\mu}^{\{k,p_{k}\}}$, i.e.,
\begin{equation}
\vec{\mu}_{i,j}(\bm{x},\bm{y})=\sum_{k=0}^{n}\sum_{{p}_{k}=1}^{{q}_{k}}{\mu}^{\{k,p_{k}\}}\cdot\vec{v}^{\{k,p_{k}\}}.
\end{equation}

For the ordered factorization of $q=q_0\times q_1\times\cdots\times q_n$,
as $e^{\{k,{p}_{k}\}}\in\Z_{4}$, ${v}^{\{k,{p}_{k}\}}\in\mathbf{Z}_{2}$ ($ k\in\mathcal{K}$), 
$\left({d_{0}^{\{k,{p}_{k}\}}}, {d_{1}^{\{k,{p}_{k}\}}}, {d_{2}^{\{k,{p}_{k}\}}}\right)\in\Z_{4}^{3}$ ($ k\notin\mathcal{K}$) runs for all their possible values, the number of vectorial offset is 
\begin{equation}
\#\{\vec{\mu}_{i,j}(\bm{x},\bm{y})\}=\prod_{k=0}^{m}\prod_{{p}_{k}=1}^{{q}_{k}}
\#\{{\mu}_{i,j}^{\{k,p_{k}\}}\}=\prod_{k\in\mathcal{K}}8^{{q}_{k}-1}\cdot\prod_{k\notin\mathcal{K}}16^{{q}_{k}-1}.
\end{equation}

\subsection{Enumeration of Affine Offset}

In this section, the constructed arrays in Corollary \ref{Coro: NS_GCA_QAM} are of size  
$\bm{2}^{(n)}\times\bm{8}^{(r)}$.
Thus the affine function ${l}(\bm{x},\bm{y})$ can  be given by
\begin{equation}\label{eq: l(x,y)}
	{{l}(\bm{x},\bm{y})}
	=\sum_{k=1}^{n}{c}_{k}{x}_{k}+\sum_{k\in\mathcal{K}}{c}'_{k}{y}_{k}+{c}_{0},
\end{equation}
where ${c}_{k}\in\Z_{4}$ ($0\le{k}\le{n}$), ${c}'_{k}\in\Z_{4}$ ($k\in\mathcal{K}$),
which determines
\begin{equation}
	\#\{{l}(\bm{x},\bm{y})\}=
	4^{n+r+1}.
\end{equation}

\subsection{Enumeration of Projections}

The projected sequences form  Golay sequence matrices over $4^{q}$-QAM of sequence length $2^{n+3r}$
by restricting $({x}_1,\dots,{x}_{n},{y}_{k_1},{y}_{k_2},\dots,{y}_{k_r})$ to be mixed radix representation of $y$ based on radixes 
$(2,\dots,{2},8,8,\dots,8)$ and permutation $\pi$.
Denote the number of projections from array to sequences based on the permutation $\pi$ by $\#\{\pi\}$. Then
\begin{equation}
	\#\{\pi\}=(n+r)!.
\end{equation}

Similar the the analysis of QPSK case in \cite{Array2}, $
\left\{{f}_{i,j}^{(p)}(\bm{x},\bm{y})\middle|0\leq{i,j}\leq1;0\leq{p}<q\right\}$ is invariant under the mapping 
\begin{align*}
&{x}_{k}\mapsto{x}_{n+1-k},\quad{c}_{k}\mapsto{c}_{n+1-k} \qquad\text{for}\quad 1\leq{k}\leq{n};\\
&\left(d_{0}^{\{k,{p}_{k}\}},d_{1}^{\{k,{p}_{k}\}},d_{2}^{\{k,{p}_{k}\}}\right)\mapsto
\left(d_{0}^{\{n+1-k,{p}_{n+1-k}\}},d_{2}^{\{n+1-k,{p}_{n+1-k}\}},d_{1}^{\{n+1-k,{p}_{n+1-k}\}}\right)\quad\text{for}\quad 1\leq{k}\leq{n};\\
&{s}_{i,j}^{\{k\}}(y_{k})\mapsto{s}_{1-i,1-j}^{\{n-k\}}(y_{n-k}),\quad{c}'_{k}\mapsto{c}'_{n-k}\qquad\text{for}\quad {k}\in\mathcal{K};\\
&\left({v}^{\{k,{p}_{k}\}},{e}^{\{k,{p}_{k}\}}\right)\mapsto\left({v}^{\{n-k,{p}_{n-k}\}},{e}^{\{n-k,{p}_{n-k}\}}\right), {R}\mapsto{L},{L}\mapsto{R}\qquad\text{for}\quad {k}\in\mathcal{K};
\end{align*}
which will cause a repetition.

So for the fixed mixed radix representation of $0\leq{p}<q$ based on the ordered factorization $q=q_0\times q_1\times\cdots\times q_n$, (this is why we call it a lower bound of enumeration,) the total number of derived sequences is 
\begin{equation}
\begin{split}
\#\left\{{f}_{i,j}^{(p)}(\bm{x},\bm{y})|0\leq{i,j}\leq1;0\leq{p}<q\right\}&=\frac{1}{2}\#\{f_{i,j}^{(0)}(\bm{x},\bm{y})\}\cdot\#\{\vec{\mu}_{i,j}(\bm{x},\bm{y})\}\cdot\#\{{l}(\bm{x},\bm{y})\}\cdot\#\{\pi\}\\
&=\frac{1}{2}\binom{n+1}{r}8^{r}\cdot\prod_{k\in\mathcal{K}}8^{{q}_{k}-1}\cdot\prod_{k\notin\mathcal{K}}16^{{q}_{k}-1}\cdot
4^{n+r+1}\cdot(n+r)!.
\end{split}
\end{equation}

\section{Concluding Remarks}\label{Sec: Concluding Remarks}
In this paper, we generalized the  multi-dimensional (or array) construction approach method into $4^{q}$-QAM constellation. 
Our work partly solved the open questions left by\cite{Array2}:

{\em How can the three-stage construction process of this paper be used to simplify or extend known results on the construction of Golay sequences in other contexts, such as 16-QAM modulation, a ternary alphabet $\{1, 0,−1\}$, or quaternary sequences whose length is not a power of $2$? }

There are some points should be noticed:
\begin{itemize}
\item [1]
We use the Para-unitary (PU) matrix to construct the Golay array pairs in terms of generating functions and extract corresponding function forms. 
The QPSK case is given in Theorem \ref{thm: PSK_CCA}, which is in agreement with the result in  \cite[Theorem 7, Corollary 10]{Array2}.
The QAM case is given in Theorems \ref{thm: QAM_CCA}. In particular, if $r=0$, i.e., no ``cross-over" PU matrices.
This makes the process of constructing array over QAM constellation more clear.
	\item[2]
	The construction over $4^{q}$-QAM constellation is based on the factorization of $q$ which is presented as the weighted sum of generating matrices over QPSK.
	The final results are presented in $4$-ary $q$-dimensional functions.
	\item[3]
	We use the mixed radix representation to simplify the process of taking projections of the resulting Golay array (pairs) to lower dimensions.
\end{itemize}

%
%


\end{document}